\documentclass[journal,twocolumn,10pt,letter]{IEEEtran}
\usepackage{lipsum,graphicx,multicol}
\usepackage{colortbl}
\usepackage{multicol}
\usepackage{lipsum}
\usepackage{color}
\usepackage[cmex10]{amsmath}
\usepackage{amsthm}
\usepackage{amssymb}
\usepackage{mathrsfs}
\usepackage{mathtools}
\usepackage{amsbsy}
\usepackage[colorlinks=true,bookmarks=false,citecolor=blue,urlcolor=blue]{hyperref}
\usepackage{xcolor}
\usepackage{mathrsfs}
\usepackage{setspace}
\usepackage{float}
\usepackage{graphicx}
\usepackage{pstool}
\usepackage{lettrine}
\usepackage{newclude}
\usepackage[normalem]{ulem}
\usepackage{latexsym}
\usepackage{algpseudocode}
\usepackage{framed,enumitem}

\usepackage[linesnumbered,ruled,vlined]{algorithm2e}

\usepackage{multirow}
\usepackage{amsmath}
\usepackage{amsmath,bm}
\usepackage{wasysym}
\usepackage{mathrsfs}
\usepackage{amsmath,cite,amsfonts,amssymb,color}
\usepackage{gensymb}

\usepackage{balance}

\ifCLASSOPTIONcompsoc
\usepackage[caption=false,font=normalsize,labelfon
t=sf,textfont=sf]{subfig}
\else
\usepackage[caption=false,font=footnotesize]{subfig}
\fi

\allowdisplaybreaks

\newtheorem{theorem}{Theorem}

\graphicspath{{figures/}}
\allowdisplaybreaks

\newcommand{\RNum}[1]{\uppercase\expandafter{\romannumeral #1\relax}}

\newcounter{mytempeqncnt3}

\algnewcommand\INPUT{\item[\textbf{Input:}]}
\algnewcommand\OUTPUT{\item[\textbf{Output:}]}
\algnewcommand\One{\item[\textbf{1.}]}
\algnewcommand\Two{\item[\textbf{2.}]}
\algnewcommand\Three{\item[\textbf{3.}]}
\algnewcommand\VarOne{\item[~~~~{1.1.}]}
\algnewcommand\VarTwo{\item[~~~~{1.2.}]}
\algnewcommand\VarThree{\item[~~~~{1.3.}]}
\algnewcommand\VarFour{\item[~~~~{1.4.}]}
\algnewcommand\VarFive{\item[~~~~{1.5.}]}
\algnewcommand\VarOneTwo{\item[~~~~{3.1.}]}
\algnewcommand\VarTwoTwo{\item[~~~~{3.2.}]}
\algnewcommand\VarThreeTwo{\item[~~~~{3.3.}]}
\algnewcommand\Four{\item[\textbf{4.}]}
\algnewcommand\VarFourOne{\item[~~~~{4.1.}]}
\algnewcommand\VarFourTwo{\item[~~~~{4.2.}]}
\algnewcommand\VarFourTwoOne{\item[~~~~~~{4.4.1.}]}
\algnewcommand\VarFourTwoTwo{\item[~~~~~~{4.4.2.}]}
\algnewcommand\VarFourThree{\item[~~~~{4.3.}]}
\algnewcommand\VarFourFour{\item[~~~~{4.4.}]}
\algnewcommand\VarFourFive{\item[~~~~{4.5.}]}
\algnewcommand\Five{\item[\textbf{5.}]}

\SetKwInput{KwInput}{Input}  
\SetKwInput{KwOutput}{Output}

\begin{document}


\title{Cognitive RF--FSO Fronthaul Assignment in Cell-Free and User-Centric mMIMO Networks}

\author{\IEEEauthorblockN{Pouya Agheli\thanks{P. Agheli, M. J. Emadi, and H. Beyranvand are with the Department of Electrical Engineering, Amirkabir University of Technology (Tehran Polytechnic), Tehran, Iran (E-mails: \{pouya.agheli, mj.emadi, beyranvand\}@aut.ac.ir).},
Mohammad Javad Emadi, and
Hamzeh Beyranvand
}} 


\maketitle

\begin{abstract}
\textcolor{black}{
Cell-free massive MIMO (CF-mMIMO) network and its user-centric (UC) version are considered as promising techniques for the next generations of wireless networks. However, fronthaul and backhaul assignments are challenging issues in these networks. In this paper, energy efficiencies of uplink transmission for the CF- and UC-mMIMO networks are studied, wherein access points (APs) are connected to aggregation nodes (ANs) through radio frequency (RF) and/or free-space optic (FSO) fronthauls, and the ANs are connected to a central processing unit via fiber backhauls. The achievable data rates are derived by taking into account the effects of hardware non-ideality at the APs and ANs, FSO alignment and weather conditions. To have a robust and energy-efficient network, especially in the presence of FSO misalignment and adverse weather conditions, firstly, a cognitive RF--FSO fronthaul assignment algorithm is proposed at the cost of sharing the available RF bandwidth between the access and fronthaul links. Then, optimal power allocations at the users and APs are investigated, and two analytical approaches are proposed to solve the non-convex optimization problem. Through numerical results, we have discussed how utilizing the cognitive RF--FSO fronthaul assignment achieves higher energy efficiency compared to that of FSO-only, RF-only, or simultaneously using RF and FSO fronthaul links, e.g., achieving up to $198\%$ higher energy efficiency under unfavorable conditions. Moreover, the effects of FSO misalignment, weather conditions, and power allocations on the performances of the CF- and UC-mMIMO networks are discussed. 
}
\end{abstract}


\IEEEpeerreviewmaketitle

\section{Introduction}
\lettrine{K}{ey} features of novel technologies for upcoming wireless networks are to answer the immediate increase of mobile users with high data rates, low latency, and energy-efficient connection requirements. To this end, a massive multiple-input multiple-output (mMIMO) network has been proposed to provide noticeable improvements in spectral and energy efficiencies by applying near-optimal linear processing thanks to the weak law of large numbers \cite{marzetta2010noncooperative}. In the mMIMO network, it is sufficient to acquire channel state information (CSI) at base stations, and due to the channel hardening and uplink/downlink reciprocity in a time-division duplexing manner, each user only needs to know the statistical averages of effective channels. Thus, there is no pilot transmission overhead in the downlink \cite{ marzetta2016fundamentals}. Various modifications of the mMIMO network have been introduced in the literature \cite{marzetta2016fundamentals,khormuji2015generalized,ngo2017cell,kabiri2017optimal}. In particular, a cell-free mMIMO (CF-mMIMO) network, which inherits the advantages of both distributed and mMIMO networks, has been introduced to serve users with almost the same high spectral efficiency and coverage support by applying the maximum-ratio detection at access points (APs) \cite{ngo2017cell}. In the CF-mMIMO setup, the APs simultaneously serve all user equipment (UEs). Since the access links between the APs and UEs are shorter, and the UEs experience almost-surely better channels, the CF-mMIMO network provides higher data rates and wider coverage support than the cellular one. In this case, \cite{ngo2017cell} has shown that the CF-mMIMO network outperforms the small-cell one with about $5$-fold and $10$-fold improvement in $95\%$-likely per-user throughput, under uncorrelated and correlated shadowing, respectively.

The CF-mMIMO networks have been investigated from various perspectives. A max-min transmission power control mechanism that ensures uniformly good service within a coverage area has been proposed in \cite{ngo2017cell}. Similarly, a power control mechanism for pilot transmission has been studied in \cite{mai2018pilot}, which minimizes the channel estimation's mean-squared error. In addition, a user-centric mMIMO (UC-mMIMO) network, wherein each user merely connects to its nearby APs, has been introduced and analyzed in \cite{buzzi2017cell}. It has been shown that the UC-mMIMO network provides higher per-user data rates in comparison to the CF-mMIMO one with less fronthaul and backhaul overhead \cite{buzzi2017user} and \cite{buzzi2019user}. Moreover, \cite{ngo2017total} and \cite{yang2018energy} have analyzed the energy efficiency and total power consumption models at the APs and \emph{ideal} backhaul links in the CF-mMIMO network. Furthermore, in \cite{masoumi2019performance} and \cite{bashar2019max}, spectral and energy efficiencies have been analyzed for \emph{limited-capacity} fronthaul CF-mMIMO networks, and the effects of quantization on the network's performance have been studied.

On the other hand, the fronthaul and backhaul links are mainly deployed by using two well-known optical technologies; fiber and free-space optic (FSO). The fiber communication technology provides high data rates and low path-loss, which comes at the disadvantages of high deployment costs and digging problems. Conversely, the FSO technology offers high enough data rates with much lower deployment cost, rapid setup time, easy upgrade, flexibility, and freedom from spectrum license regulations. However, it comes at the expense of some drawbacks such as pointing error, the requirement of a good-enough line-of-sight (LOS) connection, and sensitivity to weather conditions \cite{hassan2017statistical,khalighi2014survey,kaushal2016optical}. To conquer the outage problem of FSO links in adverse weather conditions, combined radio frequency (RF) and FSO, namely hybrid RF--FSO, and buffer-aided RF--FSO solutions have been suggested \cite{douik2016hybrid,touati2016effects,chen2016multiuser,jamali2016link,najafi2017optimal, hassan2017statistical }. \cite{najafi2017c} has presented a cloud-radio access network (C-RAN) with RF access and hybrid RF--FSO fronthaul links in which RF-based fronthaul and access links exploit the same frequency band with an optimized time-division mechanism. Furthermore, optimal power allocations and FSO fronthaul selections in cloud small-cell millimeter-wave networks have been proposed in \cite{hassan2019joint}. 

To have a lower-complexity signal processing and skip complete channel decoding and re-encoding at a relay node, instead of utilizing the well-known decode-and-forward relaying scheme, one can use the classical amplify-and-forward (AF) scheme. Besides, for converting the received signal at a relay to another domain, e.g., radio-over-FSO or FSO-over-\emph{fiber}, a practical clipping model could be used \cite{ahmed2018c, ansari2013impact,anees2015performance,soleimani2015generalized}. In practice, communications suffer from various analog and digital non-idealities, such as phase noise, impedance mismatch, I/Q imbalance, and quantization \cite{masoumi2019performance}. Despite investing in high-quality devices and applying more complex signal processing, some non-negligible errors still remain. These errors are called hardware impairment (HI) \cite{studer2010mimo}. Therefore, to propose a cost-efficient mMIMO network, analyzing its performance, subject to the HI model, has gained research interest \cite{bjornson2014massive,zhang2018performance,bjornson2015massive,zhang2017spectral}.

To the best of our knowledge, the CF- and UC-mMIMO networks' performances with two layers of optical fronthaul and backhaul links, subject to different hardware models, FSO alignment and weather conditions, have not been investigated in the literature. In this paper, we study the uplink transmissions of the CF- and UC-mMIMO wireless networks in which RF and FSO fronthaul links connect distributed APs to aggregation nodes (ANs), and fiber backhaul links connect the ANs to a central processing unit (CPU). The clipping and HI models are also used to model the reformations of the radio and optical signals at the APs and ANs. The main contributions of the paper are summarized as follows.
\begin{itemize}
    \item Closed-form uplink achievable data rates are derived by employing maximum-ratio-combining (MRC) and use-and-then-forget (UatF) techniques for the CF- and UC-mMIMO networks.
    \item Optimal power allocations at the UEs and APs are proposed to maximize energy efficiencies of the CF- and UC-mMIMO networks, subject to maximum transmission and consumed powers at the UEs and APs. To overcome the non-convexity of the optimization problem, we suggest two solutions based on the geometric programming (GP) and weighted minimum mean-square error (W-MMSE) approaches.
    \item A cognitive fronthaul assignment algorithm is proposed to maximize the energy efficiencies of the CF- and UC-mMIMO networks under FSO misalignment and adverse weather conditions. To this end, an AP cognitively decides to transmit data over its FSO-only, RF-only, or RF--FSO fronthaul. For the case of using RF bands in the fronthaul links, e.g., RF-only and RF--FSO, the RF bandwidth is shared between the access and fronthaul. 
    \item Through numerical results, the performances of the CF- and UC-mMIMO networks and the advantages of the optimal power allocations and cognitive fronthaul assignments are investigated. It is shown that the UC-mMIMO network outperforms the CF-mMIMO one from the spectral and energy efficiency viewpoints. It is also concluded that the optimally allocating the UEs' and APs' transmission powers improves the networks' performances compared to full power allocations. Likewise, the networks with cognitively-assigned fronthaul links offer higher energy efficiencies than that of with FSO-only, RF-only, and simultaneously using both RF and FSO links (RF\&FSO), in unfavorable FSO alignment and weather conditions.
\end{itemize}

\textit{Organization}: Section \ref{Sec:Sec2} introduces the CF- and UC-mMIMO networks with FSO and RF fronthaul and fiber backhaul links.
Channel training and data transmission are represented in Section \ref{Sec:Sec3}. Uplink achievable data rates are derived in Section \ref{Sec:Sec4}. Section \ref{Sec:Sec5} presents the power allocations and cognitive fronthaul assignments, and Section \ref{Sec:Sec6} represents numerical results and discussions. Finally, the paper is concluded in Section \ref{Sec:Sec7}.

\textit{Notation}: $\bm{x} \! \in \! \mathbb{C}^{n \times 1}$ represents a vector in an $n$-dimensional complex space, $\mathbb{E}\{\cdot\}$ is the expectation operator, and $[\,\cdot\,]^T$ stands for the transpose. Also, $\operatorname{erf}(.)$ and $\operatorname{erfc}(.)$ indicate the error and complementary error functions, respectively. Moreover, $y \!\sim\! \mathcal{N}(m,\sigma^2)$ and $z \!\sim\! \mathcal{CN}(m,\sigma^2)$ sequentially denote real-valued and complex symmetric Gaussian random variables (RVs) with mean $m$ and variance $\sigma^2$.

\section{System Model} \label{Sec:Sec2}
We consider a distributed mMIMO network wherein $M$ number of $N$-antenna APs simultaneously serve $K$ single-antenna UEs at the same time--frequency resources. The users are distributed over a wide area with different regions, e.g., urban and rural areas, industrial estates, and crowded places, c.f. Fig.~\ref{Fig:Fig.1}. To cover the whole area, it is also assumed that there exists $A$ ANs to gather data from all the APs via RF--FSO links, namely the \emph{fronthaul} links. Every AP is equipped with a single RF fronthaul antenna and an FSO aperture for sending signals to the serving AN. Ultimately, each AN re-transmits its received signals to the CPU via a multi-core fiber (MCF) link, called the \emph{backhaul} link. To analyze the performance of the proposed network, we consider two setups; (a) cell-free network and (b) user-centric one. Despite the cell-free setup in which an AP serves all UEs, for the user-centric case, each AP only serves a portion of its adjacent UEs. Moreover, to acquire channel state information (CSI), each coherence time interval is divided into two phases; uplink channel training and data transmission, where UE--AP access and AP--AN fronthaul links are estimated at the APs and ANs, respectively.

\begin{figure}[!t]
    \centering
    \hspace{-0.465cm}
    \subfloat{
    \pstool[scale=0.202]{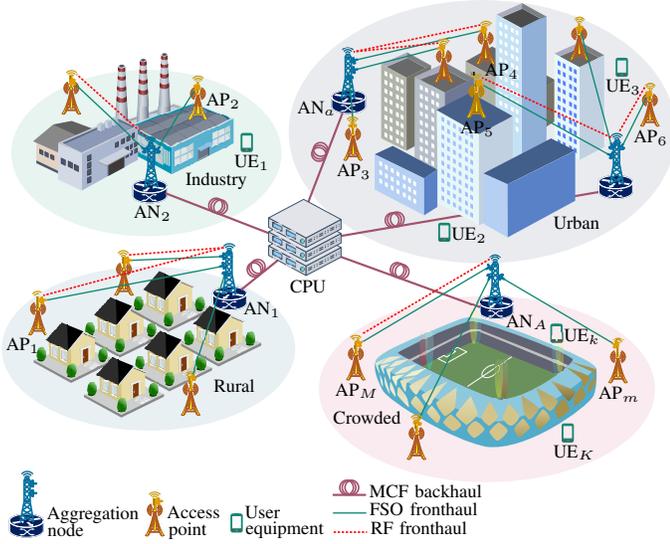}{
    \psfrag{Industry}{\hspace{-0.15cm}\scriptsize Industry}
    \psfrag{Rural}{\scriptsize \hspace{-0.15cm}  Rural}
    \psfrag{Urban}{\scriptsize \hspace{-0.15cm} Urban}
    \psfrag{Stadium}{\hspace{-0.2cm}\scriptsize Crowded}
    \psfrag{CPU}{\hspace{-0.2cm} \scriptsize CPU}
    \psfrag{Fiber}{\hspace{-0.8mm}\scriptsize MCF backhaul}
    \psfrag{mmWave}{\hspace{-0.6mm}\scriptsize RF fronthaul}
    \psfrag{FSO}{\hspace{-0.8mm}\scriptsize FSO fronthaul}
    \psfrag{AP1}{\hspace{-0.15cm}\scriptsize $\text{AP}_1$}
    \psfrag{AP2}{\hspace{-0.05cm}\scriptsize $\text{AP}_2$}
    \psfrag{AP3}{\hspace{-0.05cm}\scriptsize $\text{AP}_3$}
    \psfrag{AP4}{\hspace{0.33cm}\scriptsize $\text{AP}_4$}
    \psfrag{AP5}{\hspace{-0.05cm}\scriptsize $\text{AP}_5$}
    \psfrag{AP6}{\hspace{-0.05cm}\scriptsize $\text{AP}_6$}
    \psfrag{APm}{\hspace{-0.05cm}\scriptsize $\text{AP}_m$}
    \psfrag{APM}{\hspace{-0.05cm}\scriptsize $\text{AP}_M$}
    \psfrag{AN1}{\hspace{0cm}\scriptsize $\text{AN}_1$}
    \psfrag{AN2}{\hspace{-0.05cm}\scriptsize $\text{AN}_2$}
    \psfrag{ANa}{\hspace{-0.33cm} \scriptsize $\text{AN}_a$}
    \psfrag{ANA}{\hspace{-0.1cm}\scriptsize $\text{AN}_A$}
    \psfrag{UE1}{ \hspace{-0.28cm} \scriptsize $\text{UE}_1$}
    \psfrag{UE2}{\hspace{-0.05cm}\scriptsize $\text{UE}_2$}
    \psfrag{UE3}{\hspace{-0.08cm}\scriptsize $\text{UE}_3$}
    \psfrag{UEk}{\scriptsize $\text{UE}_k$}
    \psfrag{UEK}{\hspace{-0.05cm}\scriptsize $\text{UE}_K$}
    \psfrag{Aggregation}{\hspace{0mm}\scriptsize Aggregation}
    \psfrag{node}{\hspace{0.1mm}\scriptsize node}
    \psfrag{Access}{\hspace{0cm}\scriptsize Access}
    \psfrag{point}{\hspace{0.2mm}\scriptsize point}
    \psfrag{User}{\hspace{0cm}\scriptsize User}
    \psfrag{equipment}{\hspace{0.1mm}\scriptsize equipment}
    }}
    \caption{The mMIMO network with RF--FSO fronthauls and MCF backhauls.}
    \label{Fig:Fig.1}
\end{figure}
\begin{figure}[t!]
    \centering
    \pstool[scale=0.47]{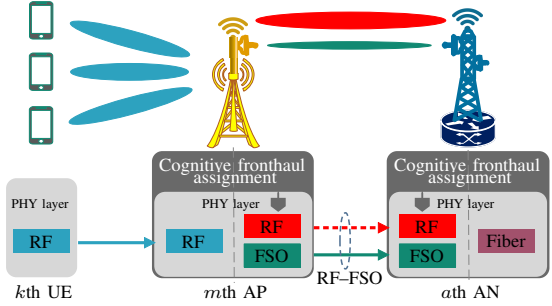}{
    \psfrag{RF}{\scriptsize \hspace{-1.6mm} RF}
    \psfrag{FSO}{\scriptsize \hspace{-1.8mm} FSO}
    \psfrag{RF/FSO}{\scriptsize \hspace{-2.1mm} RF--FSO}
    \psfrag{Access}{\scriptsize \hspace{-1.8mm} Access}
    \psfrag{Fronthaul}{\scriptsize \hspace{-2.5mm} Fronthaul}
    \psfrag{Fiber}{\scriptsize \hspace{-2mm} Fiber}
    \psfrag{UE}{\scriptsize \hspace{-3mm} $k$th UE}
    \psfrag{APm}{\scriptsize \hspace{-3mm} $m$th AP}
    \psfrag{ANa}{\scriptsize \hspace{-3mm} $a$th AN}
    \psfrag{PHYLayer}{\tiny \hspace{-1mm} PHY layer}
    \psfrag{Cognitivefronthaul}{\scriptsize \hspace{-4.8mm} \textcolor{white}{Cognitive fronthaul}}
    \psfrag{assignment}{\scriptsize \hspace{-2.2mm} \textcolor{white}{assignment}}
    }
    \caption{The signal flow over access and cognitively-assigned fronthaul links.}
    \label{Fig:Fig.Ass}
\end{figure}

Before investigating achievable data rates and optimization problems, we study data flow and provide channel and hardware models in the following subsections.

\subsection{Fronthaul Signal Flow}
Considering the \emph{alignment} precision and \emph{weather} conditions as two main factors affecting the FSO performance, each AP decides to forward the received signal over its FSO-only, RF-only, or RF--FSO fronthaul, c.f. Fig.~\ref{Fig:Fig.Ass}, based on a \emph{cognitive fronthaul assignment} algorithm. To this end, we consider four weather conditions--\emph{clear}, \emph{rainy}, \emph{snowy}, and \emph{foggy}. Also, we suggest three FSO alignment conditions--\emph{poor}, \emph{moderate}, and \emph{good} alignments. In the case of RF-only and RF--FSO fronthauling, the available RF bandwidth is shared between the access and fronthaul links due to its high cost, and the frequency-division multiplexing (FDM) method is applied. Let us consider the RF bandwidth as $\text{BW}$, and define $\text{BW}_0$ and $\text{BW}_m$ as the bandwidths of the access layer and the $m$th AP's RF fronthaul, respectively.  So, we have $\text{BW}_{0}\!=\!\underset{a=1,2,...,A}{\text{min}}\{\text{BW}_{0_a}\!\}$, where $\text{BW}_{0_a}\!=\!\text{BW}_{m}\!=\!\text{BW}/(\Delta_a\!+\!1)$, $m\!\in\! \mathcal{M}(a)$\footnote{Since the received signals at an AP are analogy reformed and forwarded, its input and output RF signals' bandwidths are equal. }. Herein, $\Delta_a$ is the number of the RF links connected to the $a$th AN, for $a=1, 2, ..., A$. Thanks to beamforming and interference cancellation techniques, the fronthauls' RF bandwidths are reused per AN. Besides, $\mathcal{M}(a)$ denotes the set of the APs served by the $a$th AN, where $\sum_{a=1}^{A} \mathcal{M}(a) \!=\! M$, and $\mathcal{M}(a)\cap\mathcal{M}(a^\prime)\!=\!\varnothing$ if $a\neq a^\prime$. 

\subsection{Access Channels}
The RF access link between the $k$th UE and $m$th AP is modeled as
\begin{equation}\label{Eq:Eq1}
    \mathbf{g}_{mk} = \beta_{mk}^{1/2} \mathbf{h}_{mk},\text{~for~} m=1, 2, ..., M \text{~and~} k=1, 2, ..., K,
\end{equation}
where $\beta_{mk}$ denotes the large-scale fading, and $\mathbf{h}_{mk}\!\in\!\mathbb{C}^{N\times 1}$ represents the small-scale Rayleigh fading coefficients such that its elements are independent and identically distributed (i.i.d.) $\mathcal{CN}(0,1)$ RVs. 
\subsection{Fronthaul Channels}
Since the fronthaul link could be FSO-only, RF-only, or RF--FSO, in what follows, we discuss channel modelings of RF and FSO links.

\subsubsection{RF links} For the case of RF link between the $m$th AP and $a$th AN, we have
\begin{equation}\label{Eq:Eq2}
    I_{am}^{\text{RF}} = \beta_{am}^{1/2}h_{am},\text{~for~} m\!\in\! \mathcal{M}(a),
\end{equation}
where $\beta_{am}$ and $h_{am}$ respectively present the large-scale and the small-scale fading terms.

\subsubsection{FSO links} For the case of FSO fronthaul, we have 
\begin{equation}\label{Eq:Eq3}
    I_{am}^{\text{FSO}} = I_{l,am} I_{t,am} I_{p,am},\text{~for~} m\!\in\! \mathcal{M}(a),
\end{equation}
where $I_{l,am}$ denotes the path-loss \cite{christopoulou2019outage}
\begin{equation}\label{Eq:EqPL}
     I_{l,am} = \frac{A_r e^{-\gamma d_{am}}}{(\phi_r d_{am})^2} := I_{l,am}^\prime e^{-\gamma d_{am}},
\end{equation}
where $A_r$, $d_{am}$, $\phi_r$, and $\gamma$ denote aperture area, fronthaul length, beam divergence angle, and weather-dependent attenuation coefficient, respectively. Besides, $I_{t,am}$ is the atmospheric turbulence-induced fading, modeled via a log-normal distribution with the following probability distribution function (p.d.f) of \cite{ahmed2018c}
\begin{align}\label{Eq:Eq4} \nonumber
&f_{I_{t,am}}(I_{t,am}) =\\
&~~~\frac{1}{2I_{t,am} \sqrt{2 \pi \delta_{l,am}^2}}\,
\operatorname{exp}\!\left(\!\!-\frac{{\left(\ln(I_{t,am}) + 2\delta_{l,am}^2\right)^2}}{8\delta_{l,am}^2}\!\right)\!\!,
\end{align}
where $\delta_{l,am}^2 \!=\! 0.307 C_n^2 k_{am}^{7/6} d_{am}^{11/6}$ is the log-amplitude variance for plane waves, and $C_n^2$ denotes the index of refraction structure parameter. Herein, $k_{am}\!=\!{2 \pi}/\lambda_{am}$ in which $\lambda_{am}$ is the wavelength of the FSO link. Furthermore, $I_{p,am}$ represents the boresight pointing error with the following pdf \cite{farid2007outage}
\begin{equation}\label{Eq:EqPE}
    f_{I_{p,am}}(I_{p,am}) = \frac{\xi^2}{I_0^{\xi^2}} I_{p,am}^{\xi^2-1}, ~~ 0\leq I_{p,am}\leq I_0,
\end{equation}
where $I_0 \!=\! \left[\operatorname{erf}(v)\right]^2$, and $\xi \!=\! \frac{1}{2} w_{z_{eq}}\sigma_s^{-1}$ implies the ratio of the equivalent beam radius to the pointing error displacement standard deviation at the receiver. Moreover, $w_{z_{eq}}^2 \!=\! w_z^2 \sqrt{0.25 \pi} \operatorname{erf}(v) v^{-1} \operatorname{exp}(v^2)$ and $v \!=\! \sqrt{0.5 \pi} w_z^{-1} r_{s} $, where $w_z$ is the beam waist at distance $z$, and $r_s$ depicts the alignment-based radial distance at the detector.
Finally, the pdf of the $I_{am}^\text{FSO}$ is obtained as \cite{farid2007outage}
\begin{equation}\label{Eq:EqTotalF}
    f_{I_{am}^\text{FSO}}(I_{am}^\text{FSO}) = \frac{\xi^2 (I_{am}^\text{FSO})^{\xi^2-1}}{2 (I_0I_{l,am})^{\xi^2}} \operatorname{erfc}\! \left(\! \frac{\ln\!\left(\!{\dfrac{I_{am}^\text{FSO}}{I_0I_{l,am}}}\!\right) \!+\! \varphi_{am}}{\sqrt{8 \delta_{l,am}^2}}\!\right) \!\varphi_{am}^\prime,
\end{equation}
where $\varphi_{am} \!=\! 2 \delta_{l,am}^2 \!\left(1+2\xi^2\right)$ and $\varphi_{am}^\prime \!=\! 2 \delta_{l,am}^2 \xi^2 \!\left(1+\xi^2\right)$.

\subsection{Hardware Models}
The received signals at the APs and ANs are analogy reformed and forwarded through the fronthaul and backhaul links, respectively. Two hardware models are employed; clipping and HI models. In general, the noisy distorted signal is modeled by the Bussgang theorem, as follows \cite{demir2020bussgang}
\begin{equation}\label{Eq:Eq5}
    x_{out} = \mu x_{in} + n_d,
\end{equation}
where $\mu\!\in\!\{\mu_m,\mu_a\}$ denotes the distortion gain. Furthermore, $x_{in}\!\sim\! \mathcal{CN}(0,\delta^2)$ is input signal, and $n_d \!\sim\! \mathcal{CN}(0,\delta_d^2)$ is the distortion noise.

\subsubsection{Clipping Model}
To convert the RF signals to the FSO form at an AP, we apply a clipping model in which the received signal is firstly clipped, and then a DC bias is added to make the optical signal non-negative. Finally, the amplified signal is transmitted with a laser diode. This technique is called RF-over-FSO (RoFSO), c.f. Fig.~\ref{Fig:Fig.3}\,(a). Likewise, the same method is used for reforming the FSO and/or RF signals to match the fiber optics at an AN, sequentially named RF-over-fiber (RoF) and FSO-over-fiber (FSoF), c.f. Fig.~\ref{Fig:Fig.3}\,(b). For the clipping model, the distortion gain becomes \cite{ahmed2018c}
\begin{equation}\label{Eq:Eq6}
    \mu = \mu_0\,{\operatorname{erf}}\!\left(\!\frac{B_c}{\sqrt{2 \mu_0^2 \delta^2}}\!\right)\!,
\end{equation}
where $\mu_0$ denotes the laser's gain, $B_c$ represents the clipping level, and the variance of the distortion noise is given by \cite{ahmed2018c}
\begin{align}\label{Eq:Eq7} \nonumber
    \delta_d^2 &= B_c^2 \Big(2-\dfrac{\mu}{\mu_0}\Big) + 
     \mu_0^2 \delta^2 \Big(1 - \dfrac{\mu}{\mu_0}\Big)  \dfrac{\mu}{\mu_0}\\
     &~~~-
    \sqrt{\frac{2 B_c^2 \mu_0^2 \delta^2}{\pi}}\,
    {\operatorname{exp}}\!\left(\!-{\frac{B_c^2}{2 \mu_0^2\delta^2}}\!\right)\!\!.
\end{align}

\begin{figure}[t!]
    \centering
    \subfloat[]{
    \pstool[scale=0.35]{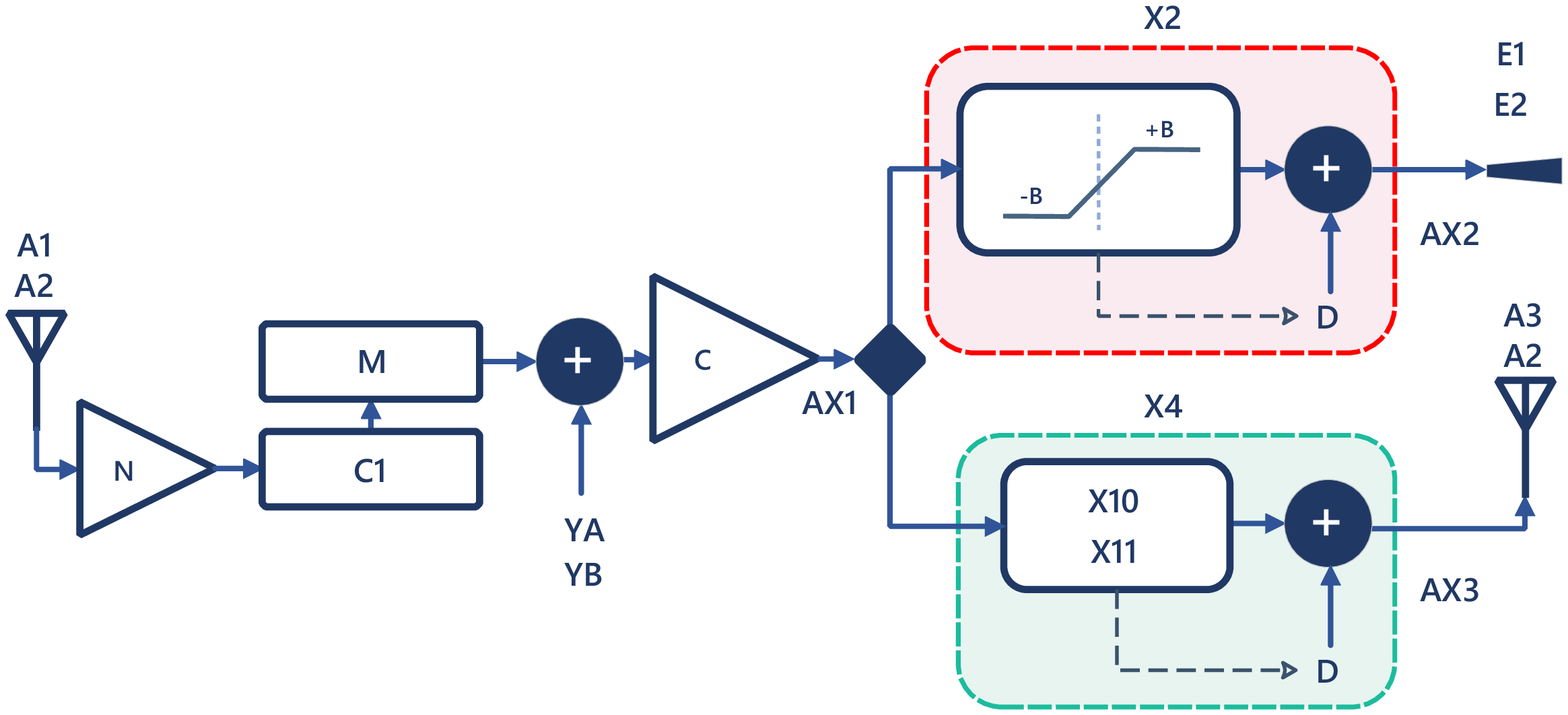}{
    \psfrag{A1}{\scriptsize \hspace{-0.15cm} RF}
    \psfrag{A2}{\scriptsize \hspace{-0.35cm} antenna}
    \psfrag{A3}{\scriptsize \hspace{-0.28cm} RF}
    \psfrag{N}{\scriptsize \hspace{-0.25cm} LNA}
    \psfrag{M}{\scriptsize \hspace{-0.48cm} IF to BB}
    \psfrag{YA}{\scriptsize \hspace{-0.42cm} Thermal}
    \psfrag{YB}{\scriptsize \hspace{-0.28cm} noise}
    \psfrag{C}{\scriptsize \hspace{-0.28cm} Driver}
    \psfrag{D}{\small \hspace{-0.15cm} $n_c$}
    \psfrag{X1}{\scriptsize \hspace{-0.17cm} $G_m$}
    \psfrag{X2}{\scriptsize \hspace{-1.27cm} Clipping model (RoFSO)}
    \psfrag{X4}{\scriptsize \hspace{-0.77cm} HI model (RoR)}
    \psfrag{C1}{\scriptsize \hspace{-0.45cm} RF to IF}
    \psfrag{AX1}{\small \hspace{-0.2cm} $x_{in}$}
    \psfrag{AX2}{\small \hspace{-0.18cm} $x_{out}^\text{FSO}$}
    \psfrag{AX3}{\small \hspace{-0.18cm} $x_{out}^\text{RF}$}
    \psfrag{X10}{\scriptsize \hspace{-0.3cm} Quality}
    \psfrag{X11}{\scriptsize \hspace{-0.24cm} factor}
    \psfrag{B}{\scriptsize \hspace{-0.2cm} $\text{B}_c$}
    \psfrag{+B}{\scriptsize \hspace{-0.2cm} +$\text{B}_c$}
    \psfrag{-B}{\scriptsize \hspace{-0.25cm} -$\text{B}_c$}
    \psfrag{E1}{\scriptsize \hspace{-0.3cm} FSO}
    \psfrag{E2}{\scriptsize \hspace{-0.55cm} laser diode}
    \psfrag{IF=}{\scriptsize IF = Intermediate frequency}
    \psfrag{BB=}{\scriptsize BB = Baseband}}
    }
    \hfill
    \subfloat[]{
    \pstool[scale=0.35]{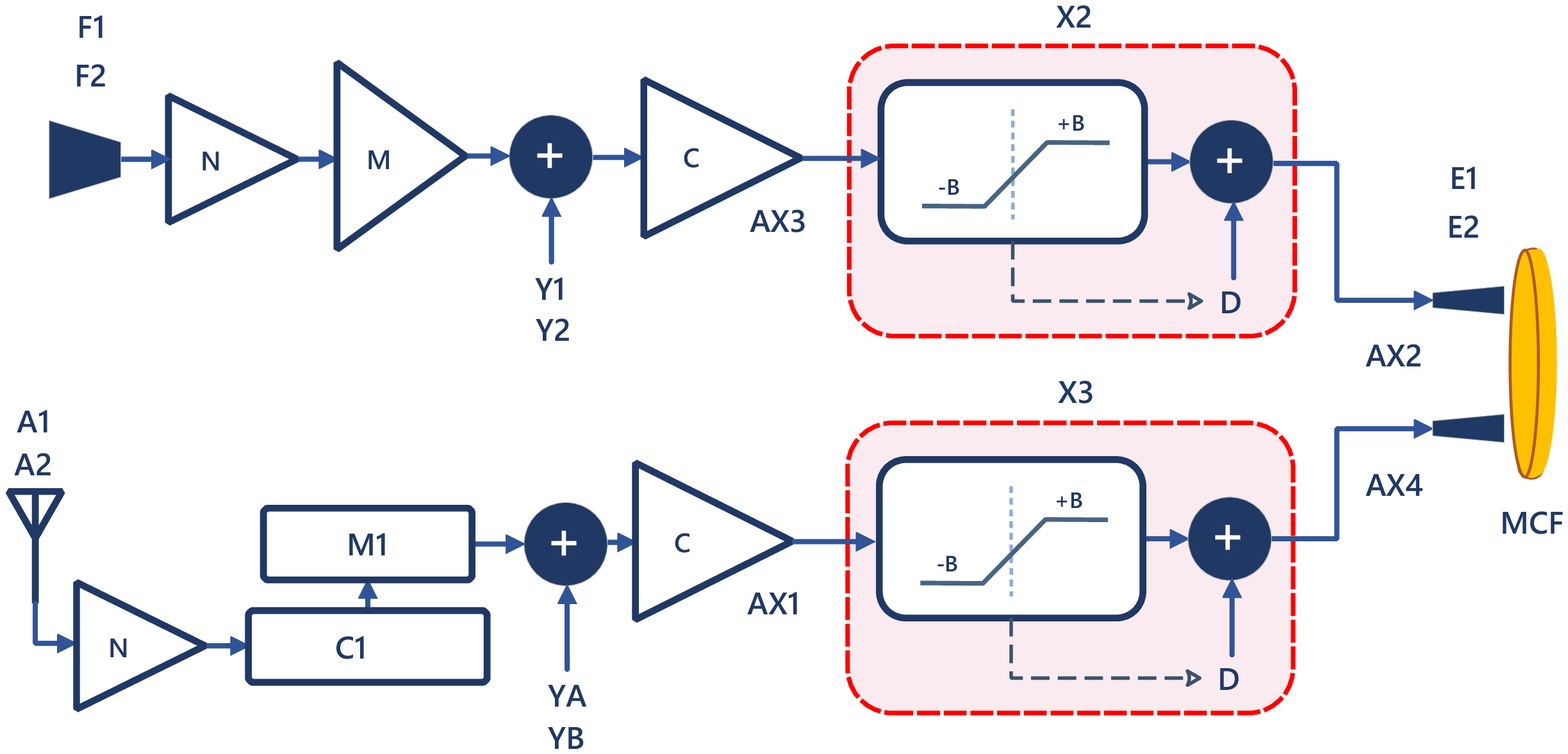}{
    \psfrag{A1}{\scriptsize \hspace{-0.28cm} RF}
    \psfrag{A2}{\scriptsize \hspace{-0.35cm} antenna}
    \psfrag{N}{\scriptsize \hspace{-0.25cm} LNA}
    \psfrag{M}{\scriptsize \hspace{-0.23cm} Amp}
    \psfrag{Y1}{\scriptsize \hspace{-0.67cm} Thermal and}
    \psfrag{Y2}{\scriptsize \hspace{-0.55cm} shot noise}
    \psfrag{YA}{\scriptsize \hspace{-0.38cm} Thermal}
    \psfrag{YB}{\scriptsize \hspace{-0.25cm} noise}
    \psfrag{C}{\scriptsize \hspace{-0.28cm} Driver}
    \psfrag{D}{\small \hspace{-0.15cm} $n_c$}
    \psfrag{X1}{\scriptsize \hspace{-0.17cm} $G_m$}
    \psfrag{X2}{\scriptsize \hspace{-1.2cm} Clipping model (FSoF)}
    \psfrag{X3}{\scriptsize \hspace{-1.15cm} Clipping model (RoF)}
    \psfrag{X4}{\scriptsize \hspace{-0.45cm} HI model}
    \psfrag{AX1}{\small \hspace{-0.23cm} $x_{in}^\text{RF}$}
    \psfrag{AX2}{\small \hspace{-0.18cm} $x_{out}^\text{FSO}$}
    \psfrag{AX4}{\small \hspace{-0.18cm} $x_{out}^\text{RF}$}
    \psfrag{MCF}{\scriptsize \hspace{-0.21cm} MCF}
    \psfrag{AX3}{\small \hspace{-0.18cm} $x_{in}^\text{FSO}$}
    \psfrag{X10}{\scriptsize \hspace{-0.3cm} Quality}
    \psfrag{X11}{\scriptsize \hspace{-0.24cm} factor}
    \psfrag{B}{\scriptsize \hspace{-0.2cm} $\text{B}_c$}
    \psfrag{+B}{\scriptsize \hspace{-0.2cm} +$\text{B}_c$}
    \psfrag{-B}{\scriptsize \hspace{-0.25cm} -$\text{B}_c$}
    \psfrag{E1}{\scriptsize \hspace{-0.4cm} Optical}
    \psfrag{E2}{\scriptsize \hspace{-0.6cm} laser diodes}
    \psfrag{F1}{\scriptsize \hspace{-0.3cm} FSO}
    \psfrag{F2}{\scriptsize \hspace{-0.75cm} photodetector}
    \psfrag{C1}{\scriptsize \hspace{-0.35cm} RF to IF}
    \psfrag{M1}{\scriptsize \hspace{-0.45cm} IF to BB}
    \psfrag{Y}{\scriptsize \hspace{-0.75cm} Thermal noise}}
    }
    \caption{The simplified hardware models at the (a) APs and (b) ANs. Here, IF stands for intermediate frequency and BB indicates baseband.}
    \label{Fig:Fig.3}
\end{figure}

\subsubsection{Hardware Impairment Model} 
Since the RF devices and signal processing are not ideal, the up-conversion of the received RF signals at an AP and forwarding them over the RF fronthaul links are subject to the HI, which is named RF-over-RF (RoR), as shown in Fig.~\ref{Fig:Fig.3}\,(a). In (\ref{Eq:Eq5}), we set $\mu\!=\!\xi^{1/2}$, where $\xi \! \in \! [0,1]$ indicates the hardware quality factor, $n_d \!\sim\! \mathcal{CN}(0,(1-\xi)\delta^2)$ is uncorrelated with the input signal, and we have $\mathbb{E}\left\{|x_{out}|^2\right\} \!=\! \mathbb{E}\left\{|x_{in}|^2\right\} \!=\! \delta^2$ 
\cite{masoumi2019performance}.

\section{Channel Training and Data Transmission} \label{Sec:Sec3}
In this section, we firstly present the uplink channel training phase to estimate the RF and FSO channel coefficients. Then, the uplink data transmissions are investigated.

\subsection{Uplink Channel Training}
In the channel training phase, the UEs and APs independently transmit channel training sequences through their access and fronthaul links, respectively. The channel coefficients of the RF access links are estimated at each AP after some mathematical computations. Likewise, the channel gains of the RF and FSO fronthaul links are acquired at the corresponding ANs. Without loss of generality, it is assumed that the CPU knows the estimated channels derived at the APs and ANs. Likewise, the estimated channels at each AN are shared with the served APs, which is required for the fronthaul assignment algorithm. 

\subsubsection{RF channel estimation at APs}
All UEs transmit their $\tau_{p}$-length mutually orthogonal pilot sequences $\sqrt{\tau_{p} \rho_p} \bm{\varphi}_k \! \in \! \mathbb{C}^{\tau_{p} \times 1}$ for $k=1,2,...,K$, where $||\bm{\varphi}_k||^2\!=\!1$ and $\tau_{p}\!\geq\!K$. Therefore, the  $N \!\times\! \tau_{p}$ received superimposed pilot matrix at the $m$th AP is given by
\begin{equation}\label{Eq:Eq8}
    \bm{y}_{p,m} = \sqrt{\tau_{p} \rho_p} \sum\limits_{k=1}^{K} \mathbf{g}_{mk} \bm{\varphi}_k^H + \bm{\omega}_{p,m},
\end{equation}
where $\rho_p$ represents pilot transmission power, and $\bm{\omega}_{p,m} \!\sim\! \mathcal{CN}(0,\sigma_{p,m}^2\mathbf{\textit{I}}_{N})$ denotes the additive noise matrix with i.i.d. elements. Afterward, the received signal is multiplied with the conjugate transpose of the $k$th UE's pilot sequence as
\begin{equation}\label{Eq:Eq9}
    \tilde{\bm{y}}_{p,mk} 
    = \bm{y}_{p,m} \bm{\varphi}_k 
    = \sqrt{\tau_{p} \rho_p} \sum\limits_{k^\prime=1}^{K} \mathbf{g}_{mk^\prime} \bm{\varphi}_{k^\prime}^H \bm{\varphi}_k
    + \bm{\omega}_{p,m} \bm{\varphi}_k.
\end{equation}
To obtain the wireless channel coefficients at the $m$th AP, a linear minimum mean-square error (LMMSE) estimation is performed as follows 
\begin{align}\label{Eq:Eq10} \nonumber
    \hat{\mathbf{g}}_{mk} &= \Big({\mathbb{E}\big\{ \mathbf{g}_{mk} \tilde{\bm{y}}_{p,mk}^H\!\big\}}\!\Big)\!\Big({\mathbb{E}\big\{\tilde{\bm{y}}_{p,mk} \tilde{\bm{y}}_{p,mk}^H\!\big\}}\!\Big)^{\!-1}\tilde{\bm{y}}_{p,mk} \\
    &= \frac{\sqrt{\tau_{p} \rho_p} \beta_{mk}}{{\tau_{p} \rho_p} \sum\limits_{k^\prime=1}^{K} \beta_{mk^\prime} |\bm{\varphi}_{k^\prime}^H \bm{\varphi}_k|^2
    + \sigma_{p,m}^2} \tilde{\bm{y}}_{p,mk} := \zeta_{mk} \tilde{\bm{y}}_{p,mk},
\end{align}
and the mean-square of the estimated channel becomes
\begin{align}\label{Eq:Eq12}
    \gamma_{mk} &=: \mathbb{E}\{\hat{\mathbf{g}}_{mk}^H \hat{\mathbf{g}}_{mk} \} = \sqrt{\tau_{p} \rho_p} \beta_{mk} \zeta_{mk}.
\end{align}

\subsubsection{FSO channel estimation at ANs} To measure the FSO channel gain between the $a$th AN and the served $m$th AP, the AP sends its $\tau_{p}$-length mutually orthogonal pilot signal $\bm{\varphi}_m$ through its FSO link, where $||\bm{\varphi}_m||^2\!=\!1$. Thus, the received signal at the $a$th AN is modeled as\footnote{We assume that the FSO detectors are sufficiently spaced to avoid optical interference and cross-talk \cite{najafi2017c}.}
\begin{equation}\label{Eq:Eq13}
    \bm{y}_{p,am} = \sqrt{\tau_{p} P_{\normalfont\text{max}}^{\normalfont\text{FSO}}} I_{am}^{\text{FSO}} \bm{\varphi}_m + \bm{\upsilon}_{p,am},
\end{equation}
where $P_{\normalfont\text{max}}^{\normalfont\text{FSO}}$ is the maximum FSO transmission power at each AP satisfying the eye-safety regulations, and $\bm{\upsilon}_{p,am} \!\sim\! \mathcal{CN}(0,(\phi_{p,am}^\text{FSO})^2\mathbf{\textit{I}}_{\tau_{p}})$ is the additive noise vector with i.i.d. elements. Similarly, by multiplying $\bm{y}_{p,am}$ with $\bm{\varphi}_m^{H}$, and applying the LMMSE\footnote{The LMMSE is a sub-optimal estimation technique for the FSO channels due to their non-Gaussian statistics.}, we have
\begin{equation}\label{Eq:Eq15}
    \hat{I}_{am}^{\text{FSO}} = 
    \frac{\sqrt{\tau_{p} P_{\normalfont\text{max}}^{\normalfont\text{FSO}}}\, \Gamma_{am}^2}{{\tau_{p} P_{\normalfont\text{max}}^{\normalfont\text{FSO}}} \Gamma_{am}^2 + (\phi_{p,am}^\text{FSO})^2} \tilde{{y}}_{p,am} := \zeta_{am}^\text{FSO} \tilde{{y}}_{p,am},
\end{equation}
where, by use of (\ref{Eq:EqTotalF}), $\Gamma_{am}^n$ is
\begin{align}\label{Eq:Eq17} \nonumber
     \Gamma_{am}^n &=: \mathbb{E}\{|I_{am}^\text{FSO}|^n\} =
    \frac{\xi^2 (I_0I_{l,am})^{n}}{\xi^2+n} \\
    &~~~\times\! \operatorname{exp}\!\bigg[2\Big(\xi^2\!+\!n\Big)^{\!2}\delta_{l,am}^2 \!-\!\Big(\xi^2\!+\!n\Big)\varphi_{am} \bigg]\varphi_{am}^\prime,
\end{align}
and the mean-square of the estimated channel becomes
\begin{align}\label{Eq:Eq18}
    \gamma_{am}^\text{FSO} &=: \mathbb{E}\{|\hat{I}_{am}^{\text{FSO}}|^2\} = \sqrt{\tau_{p} P_{\normalfont\text{max}}^{\normalfont\text{FSO}}}\, \Gamma_{am}^2 \zeta_{am}^\text{FSO}.
\end{align}

\subsubsection{RF channel estimation at ANs} 
The $m$th AP transmits $\tau_{p}$-length mutually orthogonal pilot sequence $\bm{\varphi}_m^\prime$ to the serving AN via its RF link. Hence, the  $\tau_{p} \!\times\! 1$ received pilot vector at the $a$th AN is given by
\begin{equation}\label{Eq:Eq19}
    \bm{y}_{p,am}^\prime = \sqrt{\tau_{p} P_{\normalfont\text{max}}^{\normalfont\text{RF}}} I_{am}^{\text{RF}} \bm{\varphi}_m^\prime + \bm{\upsilon}_{p,am}^\prime,
\end{equation}
where $P_{\normalfont\text{max}}^{\normalfont\text{RF}}$ denotes the maximum wireless transmission power at each AP, and $\bm{\upsilon}_{p,am}^\prime \!\sim\! \mathcal{CN}(0,(\phi_{p,am}^\text{RF})^2\mathbf{\textit{I}}_{\tau_{p}})$ represents the additive noise vector with i.i.d. components. Since $||\bm{\varphi}_m^\prime||^2\!=\!1$, by multiplying $\bm{y}_{p,am}^\prime$ with $\bm{\varphi}_m^{\prime H}$, and applying the LMMSE, we have
\begin{equation}\label{Eq:Eq21}
    \hat{I}_{am}^{\text{RF}} = 
    \frac{\sqrt{\tau_{p} P_{\normalfont\text{max}}^{\normalfont\text{RF}}}\, \beta_{am}}{{\tau_{p} P_{\normalfont\text{max}}^{\normalfont\text{RF}}} \beta_{am}
    + (\phi_{p,am}^\text{RF})^2}\tilde{{y}}_{p,am}^\prime 
    := \zeta_{am}^\text{RF} \tilde{{y}}_{p,am}^\prime,
\end{equation}
and the mean-square of the estimated channel is derived as
\begin{align}\label{Eq:Eq23}
    \gamma_{am}^\text{RF} &=: \mathbb{E}\{|\hat{I}_{am}^{\text{RF}}|^2\} = \sqrt{\tau_{p} P_{\normalfont\text{max}}^{\normalfont\text{RF}}}\, \beta_{am} \zeta_{am}^\text{RF}.
\end{align}

\subsection{Uplink Data Transmission}
After receiving the UEs' data signal in the uplink by each AP, it reforms and forwards the processed signal over its fronthaul to the dedicated AN. Then, the corresponding AN reshapes the received FSO and/or RF signals to the optical ones and delivers them to the CPU via its MCF backhaul link. Finally, users' data recovery is performed at the CPU. In what follows, the details of data recovery are presented.

\subsubsection{User-centric clustering}
In the UC-mMIMO network \cite{buzzi2017user}, the $m$th AP only supports a subset of all UEs, i.e., $\mathcal{K}(m)\!\subseteq\!\{1, 2, ..., K\}$. Consequently, $\mathcal{M}(k)$ denotes the set of APs serving the $k$th UE, which is defined as
$
\mathcal{M}(k) = \left\{m : k\in \mathcal{K}(m) \right\}\!.
$
By applying a neighborhood clustering mechanism \cite{buzzi2017cell}, the $m$th AP classifies the estimated channel coefficients in a descending order and only serves a limited number of UEs, i.e., $|\mathcal{K}(m)|\!\leq\!K$. 

\subsubsection{Data processing at APs} 
First, the received signal at the $m$th AP is given by
\begin{equation}\label{Eq:Eq25}
    \bm{y}_{u,m} = \sqrt{\rho_u} \sum\limits_{k=1}^{K} \sqrt{\eta_k} \mathbf{g}_{mk} s_k + \bm{\omega}_{u,m},
\end{equation}
where $s_k \!\sim\! \mathcal{CN}(0,1)$ denotes the information symbol of the $k$th UE weighted by power control coefficient $\eta_k \!\in\! [0,1]$. Also, $\rho_u$ represents the maximum signal transmission power, and $\bm{\omega}_{u,m} \!\sim\! \mathcal{CN}(0,\sigma_{u,m}^2\mathbf{\textit{I}}_{N})$ indicates the additive noise vector with i.i.d. elements. Then, the received signal is reformed using the RoFSO and/or RoR techniques and forwarded through FSO and/or RF fronthaul links, respectively, as follows
\begin{equation}\label{Eq:EqFSO1}
    \tilde{\bm{x}}_{u,am} = \epsilon_{am} \Big( \mu_m\bm{y}_{u,m} + \bm{n}_{d,m} \!\Big) := \epsilon_{am} {\bm{x}}_{u,am},
\end{equation}
\begin{equation}\label{Eq:EqRF1}
   \tilde{\bm{x}}_{u,am}^\prime =  \epsilon_{am}^\prime \Big( \mu_m^\prime \bm{y}_{u,m} + \bm{n}_{d,m}^\prime \!\Big) := \epsilon_{am}^\prime {\bm{x}}_{u,am}^\prime,
\end{equation}
where $\mu_m$ and $\mu_m^\prime$ denote the FSO and RF distortion gains, respectively. Herein, $\bm{n}_{d,m} \!\sim\! \mathcal{CN}(0,\delta_{d,m}^2 \mathbf{\textit{I}}_N)$ and $\bm{n}_{d,m}^\prime \!\sim\! \mathcal{CN}(0,\delta_{d,m}^{\prime 2} \mathbf{\textit{I}}_N)$ represent the distortion noises. In addition, $\epsilon_{am}\!\in\!\{0,1\}$ and $\epsilon_{am}^\prime\!\in\!\{0,1\}$ present the \emph{fronthaul assignment indices} of the FSO and RF links between the $a$th AN and the $m$th AP, $m\!\in\!\mathcal{M}(a)$, respectively. With $\epsilon_{am}\!=\!1$ and $\epsilon_{am}^\prime\!=\!0$, the signal is sent over the FSO-only fronthaul, while with $\epsilon_{am}\!=\!0$ and $\epsilon_{am}^\prime\!=\!1$, the signal is sent through the RF-only fronthaul. However, the signal is transmitted over the RF--FSO fronthaul if $\epsilon_{am}\!=\!\epsilon_{am}^\prime\!=\!1$.

\subsubsection{Data processing at ANs}
The received FSO and/or RF signals at the $a$th AN are sequentially presented as below
\begin{equation}\label{Eq:Eq26}
    \bm{y}_{u,am} = \epsilon_{am} \Big( {\bm{x}}_{u,am} I_{am}^\text{FSO} + \bm{\upsilon}_{u,am} \!\Big),
\end{equation}
\begin{equation}\label{Eq:Eq27}
    \bm{y}_{u,am}^\prime = \epsilon_{am}^\prime \Big( {\bm{x}}_{u,am}^\prime I_{am}^\text{RF} + \bm{\upsilon}_{u,am}^\prime \!\Big),
\end{equation}
Also, $\bm{\upsilon}_{u,am} \!\sim\! \mathcal{CN}(0,(\phi_{u,am}^\text{FSO})^2\mathbf{\textit{I}}_N)$ and $\bm{\upsilon}_{u,am}^\prime \!\sim\! \mathcal{CN}(0,(\phi_{u,am}^\text{RF})^2\mathbf{\textit{I}}_N)$ are the additive noise vectors at the AN with i.i.d. elements. Then, the AN applies the FSoF and RoF techniques to match the FSO and/or RF signals with the optical forms, respectively, and sends them over the connected MCF backhaul link. 

\subsubsection{Data processing at the CPU}
The CPU acquires the reshaped and noisy versions of the FSO and/or RF signal vectors sequentially as below
\begin{equation}\label{Eq:Eq28}
    \bm{r}_{u,am} = \mu_{am}\bm{y}_{u,am} + \epsilon_{am} \Big( \bm{n}_{d,am} + \bm{\nu}_\text{CPU} \!\Big),
\end{equation}
\begin{equation}\label{Eq:Eq29}
    \bm{r}_{u,am}^\prime = \mu_{am}^\prime \bm{y}_{u,am}^\prime + \epsilon_{am}^\prime \Big( \bm{n}_{d,am}^\prime + \bm{\nu}_\text{CPU} \!\Big),
\end{equation}
where $\mu_{am}$ and $\mu_{am}^\prime$ represent the FSO and RF distortion gains, respectively. Also, $\bm{n}_{d,am} \!\sim\! \mathcal{CN}(0,\delta_{d,am}^2 \mathbf{\textit{I}}_N)$ and $\bm{n}_{d,am}^\prime \!\sim\! \mathcal{CN}(0,\delta_{d,am}^{\prime 2} \mathbf{\textit{I}}_N)$ illustrate the distortion noises, and $\bm{\nu}_\text{CPU} \!\sim\! \mathcal{CN}(0,\phi_\text{CPU}^2\mathbf{\textit{I}}_N)$ denotes the additive noise at the CPU. By using (\ref{Eq:Eq25})--(\ref{Eq:Eq29}), the received signals at the CPU are rewritten as
\begin{equation}\label{Eq:Eq30}
    \bm{r}_{u,am} = \epsilon_{am} \bigg(\! J_{am} I_{am}^\text{FSO} \sum\limits_{k=1}^{K} \sqrt{\eta_k} \mathbf{g}_{mk} s_k + \bm{\Theta}_{u,am} \!\bigg),
\end{equation}
\begin{equation}\label{Eq:Eq31}
    \bm{r}_{u,am}^\prime = \epsilon_{am}^\prime \bigg(\! J_{am}^\prime I_{am}^\text{RF} \sum\limits_{k=1}^{K} \sqrt{\eta_k} \mathbf{g}_{mk} s_k + \bm{\Theta}_{u,am}^\prime \!\bigg),
\end{equation}
where $J_{am} \!=\! \sqrt{\rho_u} \mu_m \mu_{am}$, $J_{am}^\prime \!=\! \sqrt{\rho_u} \mu_m^\prime \mu_{am}^\prime$, and we have
\begin{align*}
     \bm{\Theta}_{u,am} &= \mu_m \mu_{am} I_{am}^\text{FSO} \bm{\omega}_{u,m} + \mu_{am} I_{am}^\text{FSO}\bm{n}_{d,m}\\
    &~~~ + \mu_{am} \bm{\upsilon}_{u,am} + \bm{n}_{d,am} + \bm{\nu}_\text{CPU},\\
    \bm{\Theta}_{u,am}^\prime &= \mu_m^\prime \mu_{am}^\prime I_{am}^\text{RF}  \bm{\omega}_{u,m} + \mu_{am}^\prime I_{am}^\text{RF} \bm{n}_{d,m}^\prime \\
    &~~~+ \mu_{am}^\prime \bm{\upsilon}_{u,am}^\prime + \bm{n}_{d,am}^\prime + \bm{\nu}_\text{CPU},
\end{align*}
where $\mathbb{E}\{\bm{\Theta}_{u,am}\}\!=\!\mathbb{E}\{\bm{\Theta}_{u,am}^\prime\}\!=\!0$. Besides, the covariance matrices of $\bm{\Theta}_{u,am}$ and $\bm{\Theta}_{u,am}^\prime$ are respectively denoted by  $\bm{\Omega}_{u,am}^2\!=\!{\Omega}_{u,am}^2\mathbf{\textit{I}}_{N\!\times\!N}$ and $\bm{\Omega}_{u,am}^{\prime 2}\!=\!{\Omega}_{u,am}^{\prime 2}\mathbf{\textit{I}}_{N\!\times\!N}$, where
\begin{align}\label{Eq:Eq32} \nonumber
    {\Omega}_{u,am}^2 &= \mu_m^2 \mu_{am}^2 \Gamma_{am}^2 \sigma_{u,m}^2 + \mu_{am}^2 \Gamma_{am}^2\delta_{d,m}^2\\
    &~~~ + \mu_{am}^2 (\phi_{u,am}^\text{FSO})^2 + \delta_{d,am}^2 + \phi_\text{CPU}^2,
\end{align}
\begin{align}\label{Eq:Eq33} \nonumber
    {\Omega}_{u,am}^{\prime2} &= \mu_m^{\prime 2} \mu_{am}^{\prime 2} \beta_{am}  \sigma_{u,m}^2 + \mu_{am}^{\prime 2} \beta_{am} \delta_{d,m}^{\prime 2}\\
    &~~~ + \mu_{am}^{\prime 2} (\phi_{u,am}^\text{RF})^2 + \delta_{d,am}^{\prime 2} + \phi_\text{CPU}^2.
\end{align}
One can show that all terms of the overall additive noises are mutually uncorrelated. Thus, the variances of the noises, given in (\ref{Eq:Eq32}) and (\ref{Eq:Eq33}), become equal to the sum of the variances of superimposed noises. 
\section{Achievable Rates Analysis} \label{Sec:Sec4}
The data recovery and uplink achievable data rates, for the CF- and UC-mMIMO networks, are discussed in the following subsections.

\subsection{Uplink Data Recovery}
By applying the MRC detection at the CPU, data recovery is performed for the CF- and UC-mMIMO networks.

\subsubsection{UC-mMIMO network} The CPU gathers the uplink signals from $A$ ANs, i.e., (\ref{Eq:Eq30}) and (\ref{Eq:Eq31}), and recovers the data symbol of the $k$th UE, i.e. $s_k$, by applying the following MRC technique
\begin{align}\label{Eq:Eq34}
    r_{u,k} = \!\sum\limits_{a=1}^{A} \sum\limits_{m\in \mathcal{M}_k(a)}\!\! \!\!\hat{\mathbf{g}}_{mk}^H \Big(\!(\hat{I}_{am}^{\text{FSO}})^*\bm{r}_{u,am} \!+\! (\hat{I}_{am}^{\text{RF}})^*\bm{r}_{u,am}^\prime \!\Big)\!,
\end{align}
where $\mathcal{M}_k(a) \!=\! \mathcal{M}(a) \cap \mathcal{M}(k)$.

\subsubsection{CF-mMIMO network}
For the CF-mMIMO network, the data recovery is employed as in (\ref{Eq:Eq34}), only by use of $\mathcal{M}_k(a) \!=\! \mathcal{M}(a)$. 

\subsection{Uplink Achievable Data Rates}
Through this subsection, we derive the uplink achievable data rates for the CF- and UC-mMIMO networks by employing the well-known use-and-then-forget (UatF) bounding technique \cite{marzetta2016fundamentals}.

\subsubsection{UC-mMIMO network} In the UC-mMIMO network, (\ref{Eq:Eq34}) is rewritten as
\begin{equation}\label{Eq:Eq35}
    r_{u,k} = \text{DS}_k s_k+ \underbrace{\text{BU}_k s_k
    + \sum\limits_{\substack{k^\prime=1 \\ k^\prime\neq k}}^{K} \text{IUI}_{kk^\prime} s_{k^\prime} + \text{N}_{k}}_{\text{Effective noise}},
\end{equation}
where 
\begin{itemize}
    \item $\text{DS}_k$ denotes the desired coefficient of the $k$th UE
    \begin{align}\label{Eq:Eq36} \nonumber
        \!\!\!\text{DS}_k &= \sqrt{\eta_k}\, \mathbb{E}\Bigg\{ \sum\limits_{a=1}^{A} \sum\limits_{m\in \mathcal{M}_k(a)}\!\! \epsilon_{am} J_{am} (\hat{I}_{am}^{\text{FSO}})^* I_{am}^{\text{FSO}}   \hat{\mathbf{g}}_{mk}^H \mathbf{g}_{mk}\\
        &~~~+ \sum\limits_{a=1}^{A} \sum\limits_{m\in \mathcal{M}_k(a)}\!\! \epsilon_{am}^\prime J_{am}^\prime (\hat{I}_{am}^{\text{RF}})^* I_{am}^{\text{RF}}  \hat{\mathbf{g}}_{mk}^H \mathbf{g}_{mk}\Bigg\},
        \end{align}
    \item $\text{BU}_k$ is the beamforming uncertainty coefficient of the $k$th UE due to the statistical knowledge of the CSI, i.e., UatF bounding,
    \begin{align}\label{Eq:Eq37} \nonumber
       \!\!\!\text{BU}_k &= \sqrt{\eta_k} \Bigg(\!  \sum\limits_{a=1}^{A} \sum\limits_{m\in \mathcal{M}_k(a)}\!\! \epsilon_{am}   J_{am} (\hat{I}_{am}^{\text{FSO}})^* I_{am}^{\text{FSO}}   \hat{\mathbf{g}}_{mk}^H \mathbf{g}_{mk}\\ \nonumber
        &- \mathbb{E}\Bigg\{ \sum\limits_{a=1}^{A} \sum\limits_{m\in \mathcal{M}_k(a)}\!\! \epsilon_{am} J_{am} (\hat{I}_{am}^{\text{FSO}})^* I_{am}^{\text{FSO}}   \hat{\mathbf{g}}_{mk}^H \mathbf{g}_{mk}\Bigg\}\\ \nonumber
        &+ \sum\limits_{a=1}^{A} \sum\limits_{m\in \mathcal{M}_k(a)}\!\! \epsilon_{am}^\prime J_{am}^\prime (\hat{I}_{am}^{\text{RF}})^* I_{am}^{\text{RF}}   \hat{\mathbf{g}}_{mk}^H \mathbf{g}_{mk} \\
        &- \mathbb{E}\Bigg\{\sum\limits_{a=1}^{A} \sum\limits_{m\in \mathcal{M}_k(a)}\!\! \epsilon_{am}^\prime J_{am}^\prime  (\hat{I}_{am}^{\text{RF}})^* I_{am}^{\text{RF}}   \hat{\mathbf{g}}_{mk}^H \mathbf{g}_{mk}\Bigg\}
        \!\Bigg)\!,
         \end{align}
    \item $\text{IUI}_{kk^\prime}$ represents the inter-user interference coefficient from the $k^\prime$th UE
    \begin{align}\label{Eq:Eq38} \nonumber
        \!\!\!\text{IUI}_{kk^\prime} &=  \sqrt{\eta_{k^\prime}} \Bigg(\!\sum\limits_{a=1}^{A} \sum\limits_{m\in \mathcal{M}_k(a)}\!\! \epsilon_{am}  J_{am} (\hat{I}_{am}^{\text{FSO}})^* I_{am}^{\text{FSO}}   \hat{\mathbf{g}}_{mk}^H \mathbf{g}_{mk^\prime}\\
        &~~~+ \sum\limits_{a=1}^{A} \sum\limits_{m\in \mathcal{M}_k(a)}\!\! \epsilon_{am}^\prime J_{am}^\prime (\hat{I}_{am}^{\text{RF}})^* I_{am}^{\text{RF}}   \hat{\mathbf{g}}_{mk}^H \mathbf{g}_{mk^\prime}\!\Bigg)\!,
        \end{align}
    \item $\text{N}_{k}$ indicates the overall additive effective noises
    \begin{align}\label{Eq:Eq39} \nonumber
        \!\!\!\text{N}_{k} &= \sum\limits_{a=1}^{A} \sum\limits_{m\in \mathcal{M}_k(a)}\!\! \epsilon_{am}  (\hat{I}_{am}^{\text{FSO}})^* \hat{\mathbf{g}}_{mk}^H \bm{\Theta}_{u,am}\\
         &~~~+ \sum\limits_{a=1}^{A} \sum\limits_{m\in \mathcal{M}_k(a)}\!\! \epsilon_{am}^\prime (\hat{I}_{am}^{\text{RF}})^* \hat{\mathbf{g}}_{mk}^H \bm{\Theta}_{u,am}^\prime.
    \end{align} 
\end{itemize}
One can show that all the given terms in (\ref{Eq:Eq35}) are mutually uncorrelated. Thus, from the information theoretic perspective, to analyze the worst-case scenario, we assume that the effective noise is modeled by an equivalent Gaussian random variable. Thus, the uplink achievable data rates of the $k$th UE becomes
\begin{equation}\label{Eq:Eq40}
    R_{u,k} = \log_2\!\Big(\!1 + \text{SINR}_k\!\Big)~ [\text{bps/Hz}],
\end{equation}
where the $\text{SINR}_k$ is derived in the following Theorem.

\begin{theorem}\label{Theorem:Th3} The SINR of the $k$th UE is computed as follows
\begin{align}\label{Eq:Eq42}
    \normalfont{\text{SINR}}_k \!=\!
    \dfrac{\eta_k \Bigg[ \sum\limits_{a=1}^{A} \sum\limits_{m \in \mathcal{M}_{k}(a)} \!\!\! \bigg(\! \mu_m \mathcal{A}(k) \!+\! \mu_m^\prime \mathcal{A}^\prime(k) \!\bigg) \Bigg]^{2}}{ \splitfrac{ \sum\limits_{a=1}^{A} \sum\limits_{m \in \mathcal{M}_{k}(a)} \!\bigg[ \mu_m^2 \bigg(\!\eta_k \mathcal{B}(k) \!+\! \sum\limits_{k^\prime \neq k}^{K} \eta_{k^\prime} \mathcal{C}(k,k^\prime)\!\bigg)}{ \!+\! \mu_m^{\prime 2} \bigg(\!\eta_k \mathcal{B}^\prime(k) \!+\! \sum\limits_{k^\prime \neq k}^{K} \eta_{k^\prime} \mathcal{C}^\prime(k,k^\prime)\!\bigg) \bigg] \!+\! \mathcal{D}(k)}},
\end{align}
where
\begin{align*}
    &\rule{0pt}{15pt} \mathcal{A}(k) \!=\! \sqrt{\rho_u} \epsilon_{am} \mu_{am} \gamma_{am}^\text{FSO} \gamma_{mk},\\
    &\rule{0pt}{15pt} \mathcal{A}^\prime(k) \!=\! \sqrt{\rho_u} \epsilon_{am}^\prime \mu_{am}^\prime \gamma_{am}^\text{RF} \gamma_{mk},\\
    &\rule{0pt}{15pt}\mathcal{B}(k) \!=\! \rho_u \epsilon_{am}^2 \mu_{am}^2 (\gamma_{am}^\text{FSO})^2 \Big(\! \gamma_{mk} \!+\! 2\beta_{mk} \!\Big)\gamma_{mk}, \\ 
    &\rule{0pt}{15pt}\mathcal{B}^\prime(k) \!=\! \rho_u \epsilon_{am}^{\prime 2} \mu_{am}^{\prime 2} \gamma_{am}^\text{RF} \Big(\! \gamma_{am}^\text{RF}\beta_{mk} \!+\! \beta_{am} \gamma_{mk} \!+\! \beta_{am} \beta_{mk} \!\Big), \\
    &\rule{0pt}{15pt}\mathcal{C}(k,k^\prime) \!=\!\!
    2 \rho_u \epsilon_{am}^2 \mu_{am}^2 (\gamma_{am}^\text{FSO})^2 \Big(\!  \zeta_{mk}^2 \sigma_{p,m}^2  \!+\! {\gamma_{mk}^2}/{\beta_{mk}}\!\Big) \beta_{mk^\prime}, \\
    &\rule{0pt}{15pt}\mathcal{C}^\prime(k,k^\prime) \!=\!\\
    &~~\rho_u \epsilon_{am}^{\prime 2} \mu_{am}^{\prime 2} \gamma_{am}^\text{RF}\big(\!\gamma_{am}^\text{RF} \!+\! \beta_{am}\!\big)\Big(\!  \zeta_{mk}^2 \sigma_{p,m}^2  \!+\! {\gamma_{mk}^2}/{\beta_{mk}}\!\Big) \beta_{mk^\prime}, \\
     &\rule{0pt}{15pt}\mathcal{D}(k) \!=\! \sum\limits_{a=1}^{A} \sum\limits_{m\in \mathcal{M}_k(a)} \!\! \gamma_{mk} \Big(\! \epsilon_{am}^2 \gamma_{am}^\text{FSO}  {\Omega}_{u,am}^{2}\!+\! \epsilon_{am}^{\prime 2} \gamma_{am}^\text{RF}  {\Omega}_{u,am}^{\prime2} \!\Big).
\end{align*}

\end{theorem}
\begin{proof}See Appendix \ref{App:App.1}.
\end{proof}

\subsubsection{CF-mMIMO network} By use of $\mathcal{M}_k(a) \!=\! \mathcal{M}(a)$ in (\ref{Eq:Eq42}), the $\text{SINR}_k$ of the CF-mMIMO network is derived.

\section{Resource Allocation} \label{Sec:Sec5}
In this section, we firstly propose a cognitive fronthaul assignment algorithm used at each AP. Then, the UEs' and APs' optimal transmission powers are obtained to maximize the energy efficiencies (EEs) of the CF- and UC-mMIMO networks.

\subsection{Cognitive Fronthaul Assignment}
To assign a fronthaul link, e.g., FSO-only, RF-only, and RF--FSO, for the $m$th AP, Algorithm \ref{Alg:Alg.1} is proposed by use of the estimated channels and measured FSO beam at the corresponding AN. Under different FSO alignment and weather conditions, the values of $\epsilon_{am}$ and $\epsilon_{am}^\prime$ are assigned. For measuring the received FSO beam's alignment at an AN, a quadrant photodiode (QPD) with 4-quadrant detectors (4QD) is utilized \cite{kaymak2018survey}, c.f. Fig~\ref{Fig:Fig.4}\footnote{One can show that the bounds for the $r_s$ among the \emph{ideal} alignment and \emph{misalignment} are obtained as $w_z\!-\!r_a\!\leq\!r_s\!\leq\!w_z\!+\!r_a$. Therefore, we have $1\!-\!w_z^{-1}r_a\!\leq\!w_z^{-1}r_s\!\leq\!1\!+\!w_z^{-1}r_a$.}. Based on the measured FSO beam at the quadrants of the QPD, namely $q_1$, $q_2$, $q_3$, and $q_4$, we obtain the FSO alignment condition. Next, an approximate value of the attenuation coefficient, i.e., ${\hat{\gamma}}$, is derived in (\ref{Eq:EqALgorithmPE}) to estimate the weather condition. Finally, the fronthaul assignment indices are acquired. 
\begin{figure}[t!]
    \centering
    \pstool[scale=0.63]{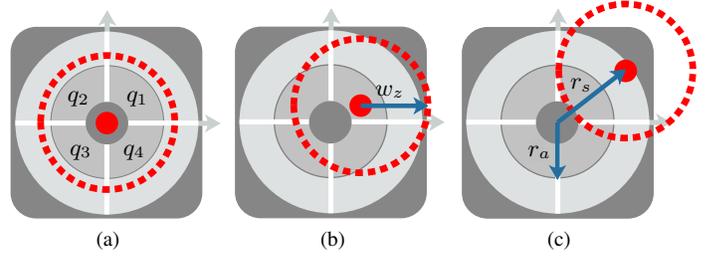}{
    \psfrag{A}{\hspace{-0.00cm}\small $q_1$}
    \psfrag{B}{\hspace{-0.05cm}\small $q_2$}
    \psfrag{C}{\hspace{-0.00cm}\small $q_3$}
    \psfrag{D}{\hspace{-0.05cm}\small $q_4$}
    \psfrag{Wz}{\hspace{-0.00cm}\small $w_z$}
    \psfrag{rs}{\hspace{-0.05cm}\small $r_s$}
    \psfrag{ra}{\hspace{-0.05cm}\small $r_a$}
    \psfrag{T1}{\hspace{-0.00cm}\footnotesize (a)}
    \psfrag{T2}{\hspace{-0.00cm}\footnotesize (b)}
    \psfrag{T3}{\hspace{-0.00cm}\footnotesize (c)}
    }
    \vspace{-3mm}
    \caption{The (a) poor, (b) moderate, and (c) good FSO alignments at the $a$th AN's QPD. The bold gray and red dashed circles represent FSO aperture and beam waist, respectively.}
    \label{Fig:Fig.4}
\end{figure}
\begin{algorithm}[t!]
\DontPrintSemicolon
    \caption{Cognitive fronthaul assignment} \label{Alg:Alg.1}
    \KwInput{FSO parameters $\xi$, $\delta_{l,am}$, $I_{l,am}^{\prime}$; attenuation factors of the supposed weather conditions, e.g., $\gamma$; fronthaul length $d_{am}$; mean-squares of the estimated channels $ \gamma_{am}^{\text{FSO}}$, $ \gamma_{am}^{\text{RF}}$; received signals at the QPD $q_1$, $q_2$, $q_3$, $q_4$. Initial parameters $r_s$, $I_0$.}
    \KwOutput{The fronthaul assignment indices: $\epsilon_{am}$, $\epsilon_{am}^\prime$.}
    \lIf{{$q_i\!\gg\! \sum_{\substack{j\!=\!1 \\ j \neq i}}^{4} q_j, \forall i$ or $q_i \!\simeq\! q_j \!\simeq\!0, \forall i\!\neq\!j$}}{poor alignment;
    }
    \lElseIf{$q_i \!\simeq\! q_j \!\gg\!0, \forall i\!\neq\!j$}{good alignment;
    }
    \lElse{moderate alignment.
    }
    Update $I_0$ using $r_s$.\\
    Using (\ref{Eq:EqPL})--(\ref{Eq:EqPE}), estimate $\hat{\gamma}$ as below
    \begin{align}\label{Eq:EqALgorithmPE}
        {\hat{\gamma}} 
        = \dfrac{1}{d_{am}}\ln\!\left(\!\dfrac{\sqrt{(\xi^2\!+\!2)}\sqrt{\gamma_{am}^\text{FSO}}}{\xi I_0 I_{l,am}^{\prime}}\!\right) \!-\! \dfrac{2\delta_{l,am}^2}{d_{am}}.
    \end{align}\\
     Compare the ${\hat{\gamma}}$ with the $\gamma$ of different weather conditions \textit{$\rightarrow$ clear, rainy, snowy, or foggy}.\\
    Assign default $\epsilon_{am}\!=\!1$, $\epsilon_{am}^\prime\!=\!0$.\\
    \lIf{$\gamma_{am}^\text{RF}\!\leq\!\gamma_{am}^\text{FSO}$}{\textbf{goto} {\scriptsize{\textbf{16}}}.}
    \lIf{poor alignment}{$\epsilon_{am}\!=\!0$, $\epsilon_{am}^\prime\!=\!1$;}
    \ElseIf{moderate alignment}{
        \lIf{snowy or foggy weather}{$\epsilon_{am}\!=\!0$, $\epsilon_{am}^\prime\!=\!1$;}
        \lElse{$\epsilon_{am}\!=\!\epsilon_{am}^\prime\!=\!1$;}}
    \Else{
        \lIf{snowy weather}{$\epsilon_{am}\!=\!\epsilon_{am}^\prime\!=\!1$;}
        \lElseIf{foggy weather}{$\epsilon_{am}\!=\!0$, $\epsilon_{am}^\prime\!=\!1$.}}
    \Return $\epsilon_{am}$, $\epsilon_{am}^\prime$; stop the process.
\end{algorithm}

\subsection{Power Allocation} Our goal is to maximize the CF- and UC-mMIMO networks' EEs by optimally allocating the UEs' and APs' transmission powers, subject to maximum power constraints. Thereafter, the optimization problem is presented by
\begin{equation}\label{Eq:Eq43}
\begin{aligned}
    &\mathcal{P}_1:
    \underset{\substack{\{\eta_k \geq 0\},\\
    \{P_m \geq 0,\, P_m^\prime \geq 0\}
    }}{\text{maximize}}~ \text{EE} \\
    &\text{subject to} \begin{cases}
    \mathcal{C}_1: \eta_k \leq 1,\hspace{0.1cm} \forall k,\\
    \mathcal{C}_2: P_m \leq P_{\normalfont\text{max}}^{\normalfont\text{FSO}},\hspace{0.1cm} \forall m,\\
    \mathcal{C}_3: P_m^\prime \leq P_{\normalfont\text{max}}^{\normalfont\text{RF}}, \hspace{0.1cm} \forall m,
    \end{cases}
\end{aligned}
\end{equation}
where $P_m$ and $P_m^\prime$ are transmission powers at the $m$th AP's FSO and RF terminals, respectively. Also, $P_{\normalfont\text{max}}^{\normalfont\text{FSO}}$ denotes the maximum FSO transmission power, satisfying the eye-safety regulations, and $P_{\normalfont\text{max}}^{\normalfont\text{RF}}$ represents the maximum RF transmission power at each AP. The EE of the network is defined as
\begin{equation}\label{Eq:Eq44}
    \text{EE} = \frac{\tau - \tau_p}{\tau} \!\left(\!\dfrac{\text{BW}_0 \sum\limits_{k=1}^{K} \log_2\!\Big(\!1 + \text{SINR}_k\!\Big)}{P_{\text{net}}}\!\right)\!,
\end{equation}
such that
\begin{align}\label{Eq:Eq45}
    P_{\text{net}} = 
    \underbrace{\rho_u \sum\limits_{k=1}^{K}\eta_k}_{\text{UEs}}
    +\! \underbrace{\sum\limits_{m=1}^{M}\!\! \Big(\!P_{c,m} \!+\! P_{\text{fh},m}\!\Big)}_{\text{APs and fronthaul links}} + \underbrace{ \sum\limits_{a=1}^{A}\Big(\!P_{c,a} \!+\! P_{\text{bh},a}\!\Big)}_{\text{ANs and backhaul links}},
\end{align}
where $P_{c,m}$ and $P_{c,a}$ denote the circuit powers at the $m$th AP and $a$th AN, respectively. Besides, $P_{\text{fh},m}$ indicates the power consumed by the $m$th fronthaul link, and $P_{\text{bh},a}$ depicts the consumed power at the $a$th backhaul link. The sum of the second and third terms in (\ref{Eq:Eq45}) is named as $P_\text{fbh}$ for next usages.

Also, $P_m$, given in $\mathcal{C}_2$, is computed as
\begin{align}\label{Eq:Eq48} \nonumber
    P_m &= \mathbb{E}\Big\{\!{\bm{x}}_{u,am}^{\!H} {\bm{x}}_{u,am}\!\Big\} \\
    &= \mu_m^2 {\rho_u} \sum\limits_{k=1}^{K} {\eta_k} \beta_{mk} \!+\! \mu_m^2 \sigma_{u,m}^2 \!+\! \delta_{d,m}^2 \!\leq\! P_\text{max}^\text{FSO}, \hspace{0.1cm} \forall m.
\end{align}
Therefore, we have
\begin{align}\label{Eq:Eq49}
    \mu_m^2 \Bigg(\! {\rho_u} \sum\limits_{k=1}^{K} {\eta_k} \beta_{mk} \!+\! \sigma_{u,m}^2 \!\Bigg) \!\leq\! P_\text{max}^\text{FSO} \!-\! \delta_{d,m}^2  \!\leq\! P_\text{max}^\text{FSO}, \hspace{0.05cm} \forall m.
\end{align}
Without loss of generality, by replacing $P_{\normalfont\text{max}}^{\normalfont\text{FSO}} \!-\! \delta_{d,m}^2$ with $P_{\normalfont\text{max}}^{\normalfont\text{FSO}}$, the feasible set becomes larger or equivalently we can assume higher maximum power. It is also worth noting that, numerical result indicates that the value of $\delta_{d,m}^2$ is negligible compared to $ P_{\normalfont\text{max}}^{\normalfont\text{FSO}}$. Similar to (\ref{Eq:Eq48})--(\ref{Eq:Eq49}), $\mathcal{C}_3$ is rewritten as
\begin{align}\label{Eq:Eq50}
    \mu_m^{\prime 2} \Bigg(\! {\rho_u} \sum\limits_{k=1}^{K} {\eta_k} \beta_{mk} \!+\! \sigma_{u,m}^2 \!\Bigg)  \!\leq\! P_\text{max}^\text{RF}, \hspace{0.1cm} \forall m.
\end{align}
According to (\ref{Eq:Eq49})--(\ref{Eq:Eq50}), finding the optimal values of $\mu_m$ and $\mu_m^\prime$ is equivalent to obtain the optimal values of $P_m$ and $P_m^\prime$. Therefore, the optimization problem $\mathcal{P}_1$ is reformulated as
\begin{equation}\label{Eq:Eq51}
\begin{aligned}
    &\mathcal{P}_2:
    \underset{\substack{\{\eta_k \geq 0\},\\
    \{\mu_m \geq 0,\, \mu_m^\prime \geq 0\}
    }}{\text{maximize}}~ \frac{\tau - \tau_p}{\tau} \!\left(\!\dfrac{\text{BW}_0}{P_0}\!\!\right)  \sum\limits_{k=1}^{K} \log_2\!\Big(\!1 + \text{SINR}_k\!\Big) \\
    &\text{subject to} \begin{cases}
    \mathcal{C}_1: \eta_k \leq 1,\hspace{0.1cm} \forall k,\\
    \mathcal{C}_2: (\ref{Eq:Eq49}), ~~
    \mathcal{C}_3: (\ref{Eq:Eq50}), ~~
    \mathcal{C}_4: P_\text{net} \leq P_0.
    \end{cases}
\end{aligned}
\end{equation}

Since the optimization problem $\mathcal{P}_2$ is a non-convex and NP-hard problem, we propose two approaches to solve it; high-SINR approximation and the equivalent W-MMSE method. Since $\frac{\tau - \tau_p}{\tau}\!\left(\!\frac{\text{BW}_0}{P_0}\!\!\right)$ is constant, it can be removed from the objective function.

\subsubsection{High-SINR approximation}
By applying the high-SINR approximation, the objective function is approximated with $\frac{\tau - \tau_p}{\tau} \!\left(\!\frac{\text{BW}_0}{P_0}\!\!\right) \log_2\!\Big( \prod\limits_{k=1}^{K}  \text{SINR}_k\!\Big)$. Moreover, because $\log_2(x) \!\leq\! x$, $\forall x\!\geq\!0$, the optimization problem is reformulated as follows
\begin{equation}\label{Eq:Eq52}
\begin{aligned}
    &\mathcal{P}_3:
    \underset{\substack{\{\eta_k \geq 0\}, \{w_k \geq 0\},\\
    \{\mu_m \geq 0,\, \mu_m^\prime \geq 0\}
    }}{\text{maximize}}~ \prod\limits_{k=1}^{K} w_k \\
    &\text{subject to} \begin{cases}
    \mathcal{C}_1: \eta_k \leq 1,\hspace{0.1cm} \forall k,\\
    \mathcal{C}_2: (\ref{Eq:Eq49}), ~~
    \mathcal{C}_3: (\ref{Eq:Eq50}), ~~
    \mathcal{C}_4: P_\text{net} \leq P_0, \\
    \mathcal{C}_5: w_k \leq \text{SINR}_k,\hspace{0.1cm} \forall k.
    \end{cases}
\end{aligned}
\end{equation}
By inserting (\ref{Eq:Eq42}) into $\mathcal{P}_3$ and further algebraic manipulations, the optimization problem becomes a geometric program (GP) since the objective function has a monomial form, and all the inequality constraints are posynomial functions, less than or equal one. Thus, $\mathcal{P}_3$ is solved by using MOSEK in CVX \cite{masoumi2019performance}.

\subsubsection{The W-MMSE method}
We can derive an equivalent form of the objective function and break the optimization problem $\mathcal{P}_2$ into subproblems to be solved sequentially \cite{8585131}. To this end, we have 
\begin{equation}\label{Eq:EqWMMSE1}
\begin{aligned}
    &\mathcal{P}_4:
    \underset{\substack{\{\varrho_k \geq 0\}, \{\mu_m \geq 0,\, \mu_m^\prime \geq 0\},\\
    \{\vartheta_k \geq 0\}, \{u_k\}
    }}{\text{minimize}}~ \sum\limits_{k=1}^{K} \Big(\!\vartheta_ke_k \!- \ln(\vartheta_k) \!\Big) \\
    &\text{subject to} \begin{cases}
    \mathcal{C}_1: \varrho^2_k \leq 1,\hspace{0.1cm} \forall k,\\
    \mathcal{C}_2: (\ref{Eq:Eq49}), ~~
    \mathcal{C}_3: (\ref{Eq:Eq50}), ~~
    \mathcal{C}_4: P_\text{net} \leq P_0,
    \end{cases}
\end{aligned}
\end{equation}
where $\varrho^2_k=\eta_k$, and $e_k$ is given in the following Theorem.
\begin{theorem}\label{Theorem:Th4} The mean-square error is computed as
\begin{align}\label{Eq:EqWMMSE2} \nonumber
    &e_k =
    |u_k|^2 \sum\limits_{k^\prime=1}^{K} \varrho^2_{k^\prime} \sum\limits_{a=1}^{A} \sum\limits_{m\in \mathcal{M}_k(a)}\!\!\! \bigg(\! \mu_m^{2} \mathcal{C}(k,k^\prime) \!+\! \mu_m^{\prime 2} \mathcal{C}^\prime(k,k^\prime) \!\bigg)\\ \nonumber
    &~+ 2 |u_k|^2 \mathfrak{Re}\Bigg\{\! \varrho^2_{k} \sum\limits_{a=1}^{A} \sum\limits_{m\in \mathcal{M}_k(a)}\!\!\! \mu_m \mu_m^\prime \mathcal{A}(k)\mathcal{A}^\prime(k)\\ 
    &~- \frac{\varrho_{k}}{|u_k|} \sum\limits_{a=1}^{A} \sum\limits_{m\in \mathcal{M}_k(a)}\!\!\! \bigg(\! \mu_m \mathcal{A}(k) \!+\! \mu_m^\prime  \mathcal{A}^\prime(k)\!\bigg)\!\!\Bigg\} \!+\! |u_k|^2 \mathcal{D}(k) \!+\!1.
\end{align}
\end{theorem}
\begin{proof} See Appendix \ref{App:App.2}.
\end{proof}
The optimization problem $\mathcal{P}_4$ is convex for each of the block variables $\varrho_k$, $\eta_k$, $\vartheta_k$, $\mu_m$, and $\mu_m^\prime$. Thus, we can apply the block coordinate descent method to solve $\mathcal{P}_4$ \cite{shi2011iteratively}. In this regard, by fixing four of the optimization variables and updating the fifth one based on the closed-form expressions obtained in Algorithm \ref{Alg:Alg.2}, the proposed algorithm sequentially converges to the stationary points $\varrho_k^*$, $\eta_k^*$, $\vartheta_k^*$, $\mu_m^*$, and $\mu_m^{\prime *}$\footnote{According to the standard steps illustrated in \cite[Theorems 1 and 3]{shi2011iteratively}, one can show that $\varrho_k^*$ is also a stationary point of $\mathcal{P}_1$.}. 
\begin{algorithm}[t!]
\DontPrintSemicolon
    \caption{Sequential solution for $\mathcal{P}_4$} \label{Alg:Alg.2}
    \KwInput{RF large-scale fadings $\beta_{mk}$,$\forall k,m$, $\beta_{am}$, $\forall m,a$; FSO intensities $\Gamma_{am}^2$, $\forall m,a$; power factors $\rho_u$, $P_\text{max}^\text{FSO}$, $P_\text{max}^\text{RF}$, $P_0$, $P_\text{fbh}$; noise variances $\sigma_{u,m}^2$, $\forall m$. Stopping accuracy $\varepsilon$. Initial coefficients $\varrho_k^{(0)}$, $\forall k$, $\mu_m^{(0)}$, $\mu_m^{\prime (0)}$, $\forall m$.}
    \KwOutput{The optimal solutions: $\varrho_k^{*}\!=\!\varrho_k^{(n)}$, $\forall k$, $\mu_m^{*}\!=\!\mu_m^{(n)}$, $\mu_m^{\prime *}\!=\!\mu_m^{\prime (n)}$, $\forall m$.}
    \textit{Iteration} $n$:\\
     Update $u_k^{(n)}$ using (\ref{Eq:EqWMMSE3}).\\
     Update $\vartheta_k^{(n)}\!=\!\big(e_k^{(n)}\big)^{\!-1}$, where $e_k^{(n)}$ is defined in (\ref{Eq:EqWMMSE4}).\\
    Update $\varrho_k^{(n)}$ as follows
    \begin{align*}
        \varrho_k^{(n)} &\!=\! \operatorname{min}\!\Bigg\{\! \tilde{\varrho}_k^{(n)},1, \sqrt{\frac{P_\text{max}^\text{FSO}\!-\!\!\sum_{m=1}^{M}\!\big(\mu_m^{(n-1)})^2 \sigma_{u,m}^2}{\rho_u \sum_{m=1}^{M}\!\big(\mu_m^{(n-1)})^2 \beta_{mk}}},\\ &\sqrt{\frac{P_\text{max}^\text{RF}\!-\!\!\sum_{m=1}^{M}\!\big(\mu_m^{\prime (n-1)})^2 \sigma_{u,m}^2}{\rho_u \sum_{m=1}^{M}\!\big(\mu_m^{\prime (n-1)})^2 \beta_{mk}}}, \sqrt{\frac{P_0 -\! P_\text{fbh}}{\rho_u}}  \Bigg\},
    \end{align*}
    where $\tilde{\varrho}_k^{(n)}$ is obtained as in (\ref{Eq:EqWMMSE5}).\\
     Update $\mu_m^{(n)}$ as below
    \begin{align*}
        \mu_m^{(n)} \!=\! \operatorname{min}\!\Bigg\{\!\tilde{\mu}_m^{(n)}, \sqrt{\frac{P_\text{max}^\text{FSO}}{\rho_u \sum_{k=1}^{K}\! \big(\varrho_k^{(n)}\big)^2 \beta_{mk} \!+\! \sigma_{u,m}^2}} \Bigg\},
    \end{align*}
    where $\tilde{\mu}_m^{(n)}$ is computed as in (\ref{Eq:EqWMMSE6}).\\
    Update $\mu_m^{\prime (n)}$ as follows
    \begin{align*}
    \mu_m^{\prime (n)} \!=\! \operatorname{min}\!\Bigg\{\!\tilde{\mu}_m^{\prime (n)}, \sqrt{\frac{P_\text{max}^\text{RF}}{\rho_u \sum_{k=1}^{K}\! \big(\varrho_k^{(n)}\big)^2 \beta_{mk} \!+\! \sigma_{u,m}^2}} \Bigg\},
    \end{align*}
    where $\tilde{\mu}_m^{\prime (n)}$ is derived in (\ref{Eq:EqWMMSE7}).\\
    \lIf{Stopping criterion $\Big|\sum\limits_{k=1}^K\!\! R_{u,k}^{(n)}-\!\sum\limits_{k=1}^K\!\! R_{u,k}^{(n-1)}\Big|\!\leq\! \varepsilon$}{stop the process;
    }
    \lElse{\textbf{goto} {\scriptsize{\textbf{9}}}.
    }
     Save the obtained solution: $\varrho_k^{(n)}$, $\forall k$, $\mu_m^{(n)}$, $\mu_m^{\prime (n)}$, $\forall m$. Set $n\!=\!n\!+\!1$, then \textbf{goto} \scriptsize{\textbf{1}}.
\end{algorithm}
\begin{figure*}[!t]
\normalsize
\setcounter{mytempeqncnt3}{\value{equation}}
\setcounter{equation}{48}
\begin{align}\label{Eq:EqWMMSE3}
    u_k^{(n)} =  \dfrac{\varrho_{k}^{(n-1)} \sum\limits_{a=1}^{A} \sum\limits_{m\in \mathcal{M}_k(a)}\!\!\! \bigg(\! \mu_m^{(n-1)} \mathcal{A}(k) \!+\! \mu_m^{\prime (n-1)}  \mathcal{A}^\prime(k)\!\bigg)}{\splitfrac{\sum\limits_{k^\prime=1}^{K} \big(\varrho^{(n-1)}_{k^\prime}\big)^2 \sum\limits_{a=1}^{A} \sum\limits_{m\in \mathcal{M}_k(a)}\!\!\! \bigg(\! \big(\mu_m^{(n-1)}\big)^2 \mathcal{C}(k,k^\prime) \!+\!  \big(\mu_m^{\prime (n-1)}\big)^2 \mathcal{C}^\prime(k,k^\prime) \!\bigg)}{ \!+\! 2\big(\varrho^{(n-1)}_{k}\big)^2 \sum\limits_{a=1}^{A} \sum\limits_{m\in \mathcal{M}_k(a)}\!\!\! \mu_m^{(n-1)} \mu_m^{\prime (n-1)} \mathcal{A}(k)\mathcal{A}^\prime(k) \!+\! \mathcal{D}(k)}}.
\end{align}
\begin{align}\label{Eq:EqWMMSE4}
    e_k^{(n)} &=
    |u_k^{(n)}|^2 u_k^{(n-1)} -2 \varrho_{k}^{(n-1)} \mathfrak{Re}\Bigg\{\!u_k^{(n)} \sum\limits_{a=1}^{A} \sum\limits_{m\in \mathcal{M}_k(a)}\!\!\! \bigg(\! \mu_m \mathcal{A}(k) \!+\! \mu_m^\prime  \mathcal{A}^\prime(k)\!\bigg)\!\!\Bigg\} +1.
\end{align}
\begin{align}\label{Eq:EqWMMSE5}
    \tilde{\varrho}_k^{(n)} =
    \dfrac{\vartheta_k^{(n)} \mathfrak{Re}\Bigg\{\! u_k^{(n)} \sum\limits_{a=1}^{A} \sum\limits_{m\in \mathcal{M}_k(a)}\!\!\! \bigg(\! \big(\mu_m^{(n-1)}\big)^{2} \mathcal{A}(k) \!+\! \big(\mu_m^{\prime (n-1)}\big)^{2}  \mathcal{A}^\prime(k)\!\bigg)\!\!\Bigg\}}{
    \splitfrac{ \sum\limits_{k^\prime=1}^{K} \vartheta_{k^\prime}^{(n)} |u_{k^\prime}^{(n)}|^2 \sum\limits_{a=1}^{A} \sum\limits_{m\in \mathcal{M}_k(a)}\!\!\! \bigg(\! \big(\mu_m^{(n-1)}\big)^{2} \mathcal{C}(k,k^\prime) \!+\! \big(\mu_m^{\prime (n-1)}\big)^{2} \mathcal{C}^\prime(k,k^\prime) \!\bigg)} {\!+\! 2 \vartheta_k^{(n)} |u_k^{(n)}|^2 \mathfrak{Re}\Bigg\{\!  \sum\limits_{a=1}^{A} \sum\limits_{m\in \mathcal{M}_k(a)}\!\!\! \big(\mu_m^{ (n-1)}\big)^{2} \big(\mu_m^{\prime (n-1)}\big)^{2} \mathcal{A}(k)\mathcal{A}^\prime(k)\!\Bigg\}}}.
\end{align}
\begin{align}\label{Eq:EqWMMSE6}
    \tilde{\mu}_m^{(n)} \!=\!\dfrac{\vartheta_k^{(n)} \varrho_{k}^{(n)} \mathfrak{Re}\Bigg\{\!u_k^{(n)} \sum\limits_{a=1}^{A} \sum\limits_{m\in \mathcal{M}_k(a)}\!\! \mathcal{A}(k)\!\Bigg\} \!-\! \vartheta_k^{(n)} |u_k^{(n)}|^2 \big(\varrho_{k}^{(n)}\big)^2 \mathfrak{Re}\Bigg\{\!  \sum\limits_{a=1}^{A} \sum\limits_{m\in \mathcal{M}_k(a)}\!\! \mu_m^{\prime (n-1)} \mathcal{A}(k)\mathcal{A}^\prime(k)\!\Bigg\}}{\sum\limits_{k^\prime=1}^{K} \vartheta_{k^\prime}^{(n)} |u_{k^\prime}^{(n)}|^2 \big(\varrho_{k^\prime}^{(n)}\big)^2 \sum\limits_{a=1}^{A} \sum\limits_{m\in \mathcal{M}_k(a)}\!\! \mathcal{C}(k,k^\prime)}.
\end{align}
\begin{align}\label{Eq:EqWMMSE7}
    \tilde{\mu}_m^{\prime (n)} \!=\!\dfrac{\vartheta_k^{(n)} \varrho_{k}^{(n)} \mathfrak{Re}\Bigg\{\!u_k^{(n)} \sum\limits_{a=1}^{A} \sum\limits_{m\in \mathcal{M}_k(a)}\!\! \mathcal{A}^\prime(k)\!\Bigg\} \!-\! \vartheta_k^{(n)} |u_k^{(n)}|^2 \big(\varrho_{k}^{(n)}\big)^2 \mathfrak{Re}\Bigg\{\!  \sum\limits_{a=1}^{A} \sum\limits_{m\in \mathcal{M}_k(a)}\!\! \mu_m^{(n)} \mathcal{A}(k)\mathcal{A}^\prime(k)\!\Bigg\}}{\sum\limits_{k^\prime=1}^{K} \vartheta_{k^\prime}^{(n)} |u_{k^\prime}^{(n)}|^2 \big(\varrho_{k^\prime}^{(n)}\big)^2 \sum\limits_{a=1}^{A} \sum\limits_{m\in \mathcal{M}_k(a)}\!\! \mathcal{C}^\prime(k,k^\prime)}.
\end{align}
\setcounter{equation}{\value{mytempeqncnt3}}
\hrulefill
\vspace*{4pt}
\end{figure*}
\setcounter{equation}{53}

\section{Numerical Results and Discussions} \label{Sec:Sec6}
Through this section, we present numerical results to highlight the advantageous of the optimally allocating the transmission powers and applying the cognitive fronthaul assignments. Conventionally, it is assumed that $K\!=\!20$ UEs and $M\!=\!200$ APs are uniformly distributed within an area of $D \!=\! 1\times1\, [\text{km}^2]$. Also, $A\!=\!4$ ANs are uniformly spread over a circle with radius (backhaul length) of $300\, [\text{m}]$ and the constant angle of $\pi/2\, [\text{rad}]$ between adjacent ANs, Fig. \ref{Fig:Fig.6}. For the UC-mMIMO, we assume that $|\mathcal{K}(m)|\!=\!10$ UEs are served by each AP. For the good, moderate, and poor FSO alignments, we initiate $w_z^{-1} r_s\!\leq\!0.8$, $0.8\!\leq\!w_z^{-1} r_s\!\leq\!1.2$, and $w_z^{-1} r_s\!\geq\!1.2$, respectively. Besides, the clear, rainy, snowy, and foggy weather conditions have the attenuation coefficients of $0.44$, $0.523$, $4.53$, and $50\, [\text{dB/km}]$, respectively \cite{ahmed2018c}.
\begin{figure}[t!]
    \centering
    \subfloat{
    \pstool[scale=0.54]{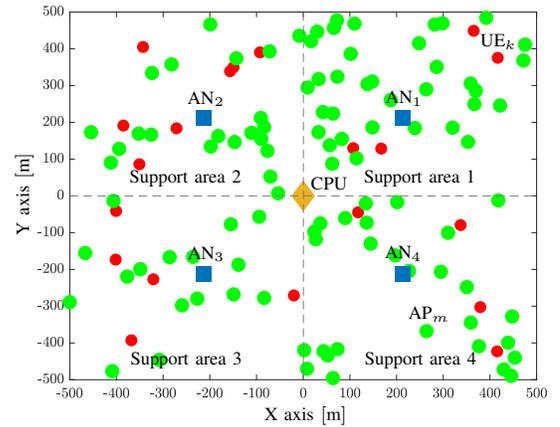}{
    \psfrag{Xaxis}{\scriptsize \hspace{-0.35cm} X axis $[\text{m}]$}
    \psfrag{Yaxis}{\scriptsize \hspace{-0.35cm} Y axis $[\text{m}]$}
    \psfrag{CPU}{\scriptsize \hspace{-0.1cm} CPU}
    \psfrag{AN1}{\scriptsize \hspace{-0.15cm} $\text{AN}_1$}
    \psfrag{AN2}{\scriptsize \hspace{-0.15cm} $\text{AN}_2$}
    \psfrag{AN3}{\scriptsize \hspace{-0.15cm} $\text{AN}_3$}
    \psfrag{AN4}{\scriptsize \hspace{-0.15cm} $\text{AN}_4$}
    \psfrag{APm}{\scriptsize \hspace{-0.15cm} $\text{AP}_m$}
    \psfrag{UEk}{\scriptsize \hspace{-0.15cm} $\text{UE}_k$}
    \psfrag{Supportarea1}{\scriptsize \hspace{-0.45cm} Support area 1}
    \psfrag{Supportarea2}{\scriptsize \hspace{-0.05cm} Support area 2}
    \psfrag{Supportarea3}{\scriptsize \hspace{-0.05cm} Support area 3}
    \psfrag{Supportarea4}{\scriptsize \hspace{-0.45cm} Support area 4}
    \psfrag{0}{\tiny $0$}
    \psfrag{100}{\tiny $100$}
    \psfrag{200}{\tiny $200$}
    \psfrag{300}{\tiny $300$}
    \psfrag{400}{\tiny $400$}
    \psfrag{500}{\tiny $500$}
    \psfrag{-100}{\tiny -$100$}
    \psfrag{-200}{\tiny -$200$}
    \psfrag{-300}{\tiny -$300$}
    \psfrag{-400}{\tiny -$400$}
    \psfrag{-500}{\tiny -$500$}
    }
    }
    \caption{A distribution shot for the UEs, APs, ANs, and CPU, wherein each AN merely serves the APs located in its coverage area.}
    \label{Fig:Fig.6}
\end{figure}
\noindent
\renewcommand{\arraystretch}{1}
\begin{table}[!t]
\begin{center}
\caption{Network parameters for numerical results.}\label{Table:Tab2}
\begin{tabular}{ | l | c | l | }
\hline
\multicolumn{1}{|c|}{\rule{0pt}{9pt} \small \centering \text{Parameter}} & \small \centering \text{Symbol} & \multicolumn{1}{c|}{\centering \small \text{Value}}\\
\hline
\hline
\footnotesize Access RF frequency & \footnotesize $f$ & \footnotesize $1.9$\, $[\text{GHz}]$\\
\hline
\footnotesize RF bandwidth & \footnotesize $\text{BW}$ & \footnotesize $40$\, $[\text{MHz}]$ \\
\hline
\footnotesize  Time coherence interval & \footnotesize $\tau$ & \footnotesize $100$\\
\hline
\footnotesize Uplink training time interval & \footnotesize $\tau_p$ & \footnotesize $20$\\
\hline
\footnotesize  Fronthaul FSO wavelength & \footnotesize $\lambda_{am}$  & \footnotesize $1550$\, $[\text{nm}]$\\
\hline
\footnotesize Index of refraction structure & \footnotesize $C_n^2$  & \footnotesize $5\!\times\! 10^{-14} [\text{{$\text{m}^{-\frac{2}{3}}$}}]$\!\!\!\\
\hline
\footnotesize  FSO displacement deviation & \footnotesize $\sigma_s$  & \footnotesize  $30$\,$[\text{cm}]$\\
\hline
\footnotesize  FSO beam radius & \footnotesize $w_z$  & \footnotesize  $2.5$\,$[\text{m}]$\\
\hline
\footnotesize  FSO receiver's radius & \footnotesize $r_a$  & \footnotesize  $10$\,$[\text{cm}]$\\
\hline
\footnotesize Clipping level & \footnotesize $B_c$ & \footnotesize $1$ \\
\hline
\footnotesize Optical RX responsibilities & \footnotesize \big\{\!$R_{\normalfont\text{FSO}}, R_{\normalfont\text{OF}}$\!\big\} & \footnotesize $\big\{0.5, 1\big\}$\\
\hline
\footnotesize AP's maximum TX powers & \footnotesize \big\{\!$P_{\normalfont\text{max}}^{\text{FSO}}, P_{\normalfont\text{max}}^{\text{RF}}$\!\big\} & \footnotesize $\big\{16, 20\big\}\, [\text{dBm}]$\\
\hline
\footnotesize UE's maximum TX powers & \footnotesize $\rho_p$, $\rho_u$ & \footnotesize $100$\,$[\text{mWatt}]$\\
\hline
\footnotesize Circuit powers & \footnotesize \big\{\!$P_{c,m}, P_{c,a}$\!\big\} & \footnotesize $\big\{0.2, 0.5\big\}\, [\text{Watt}]$ \\
\hline
\footnotesize Front/backhaul powers & \footnotesize \big\{\!$P_{\text{fh},m},P_{\text{bh},m}$\!\big\} & \footnotesize $\big\{0.1, 0.5\big\}\, [\text{Watt}]$ \\
\hline
\footnotesize Maximum available power & \footnotesize $P_0$ & \footnotesize $60\, [\text{Watt}]$ \\
\hline
\footnotesize \multirow{4}{*}{Additive noise variances} & \footnotesize $\sigma^2_{p,m}$, $\sigma^2_{u,m}$ & \footnotesize $k_B\!\cdot\! T_0\!\cdot\! \text{BW} \!\cdot\! F$\!\! \\
\cline{2-3}
\footnotesize & \footnotesize \!\!\!$(\phi_{p,am}^\text{FSO})^2$,$(\phi_{u,am}^\text{FSO})^2$ \!\!\!\!\! & \footnotesize $10^{-14}$\, $[\text{A}^2]$ \\
\cline{2-3}
\footnotesize & \footnotesize \!\!\!$(\phi_{p,am}^\text{RF})^2$,$(\phi_{u,am}^\text{RF})^2$ \!\!\!\!\! & \footnotesize $k_B\!\cdot\! T_0\!\cdot\! \text{BW}_m \!\cdot\! F$\!\!\\
\cline{2-3}
\footnotesize & \footnotesize $\phi_{\text{CPU}}^2$ & \footnotesize $10^{-14}$\, $[\text{A}^2]$\\
\hline
\footnotesize Boltzmann constant & \footnotesize $k_B$ & \footnotesize $1.381\!\times\! 10^{-23}$\!\! $[\text{J/K}]$\!\!\!\\
\hline
\footnotesize Noise temperature & \footnotesize $T_0$ & \footnotesize $290$\, $[\text{K}]$\\
\hline
\footnotesize Noise figure & \footnotesize $F$ & \footnotesize $9$\, $[\text{dB}]$\\
\hline
\end{tabular}
\medskip
\end{center}
\end{table}
For the large-scale fading, we have
$
{\beta}_{mk} \!=\! \text{PL}_{mk} \!+ \sigma_{\text{sh}}z_{mk},
$
where PL$_{mk}\, [\text{dB}]$ denotes the path-loss and $\sigma_{\text{sh}}z_{mk}$ is the shadowing with standard deviation ${\sigma}_{\text{sh}}\!=\!8\, [\text{dB}]$ and shadowing correlation factor $z_{mk}\!\sim\! \mathcal{N}(0,1)$ \cite{ngo2017cell}. Employing the conventional three-slope propagation model, we have 
\begin{equation} \label{Eq:EqNum2}
    \text{PL}_{mk} \!= \!\! \begin{cases}
    \!-L\!-\!35 {\log}_{10}(d_{mk})\text{,}&\!\! d_{mk} \!>\! d_1\\
    \!-L \!-\! 15 {\log}_{10}(d_{1}) \!-\! 20 {\log}_{10}(d_{mk})\text{,}&\!\! d_0 \!<\!  d_{mk} \!\leq\! d_1\\
    \!-L\!-\!15 {\log}_{10}(d_{1}) \!-\! 20 {\log}_{10}(d_{0})\text{,}&\!\! d_{mk} \!\leq\! d_0
\end{cases}
\end{equation}
where $d_0\!=\!10\, [\text{m}]$ and $d_1\!=\!50\, [\text{m}]$ are distance references, $d_{mk}$ denotes the distance between the $m$th AP and the $k$th UE, and $L\,[\text{dB}]$ is defined as follows
\begin{align} \label{Eq:EqNum3} \nonumber
    L &= 46.3\!+\!33.9 \, {\log}_{10}(f)\!-\!13.82\, {\log}_{10}(h_{\text{AP}})\\
    &~- (1.1\, {\log}_{10}(f)\!-\!0.7)h_{\text{UE}}\!+\!(1.56\, {\log}_{10}(f)\!-\!0.8),~~~
\end{align}
where $f$ (in $[\text{MHz}]$) represents access RF frequency, $h_{\text{AP}}\!=\!15\, [\text{m}]$ and $h_{\text{UE}}\!=\!1.65\, [\text{m}]$ are the antenna heights of an AP and a UE, respectively. 
Similarly, $\beta_{am}$ is modeled based on (\ref{Eq:EqNum2}) by replacing the $d_{mk}$ with $d_{am}$, $h_\text{AP}$ with $h_\text{AN}\!=\!30\, [\text{m}]$, $h_\text{UE}$ with $h_\text{AP}$, and so forth. The parameters used for the numerical results are summarized in Table~\ref{Table:Tab2}; otherwise, they are clearly mentioned in the paper.

\begin{figure}[t!]
    \centering
    \subfloat{
    \pstool[scale=0.57]{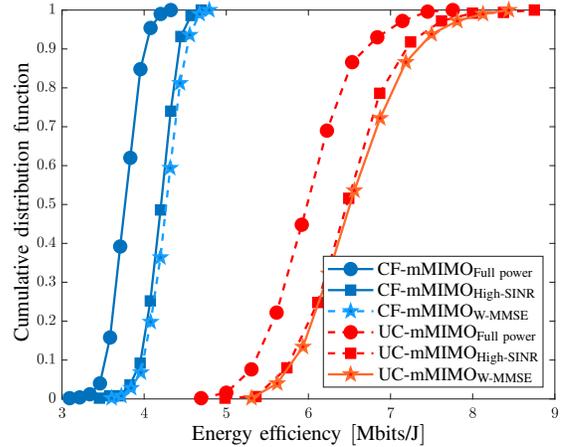}{
    \psfrag{Cumulativedistributionfunction}{\footnotesize \hspace{-0.59cm} Cumulative distribution function}
    \psfrag{Energyefficiency[Mbits/J]}{\footnotesize \hspace{-0.55cm} Energy efficiency $[\text{Mbits/J}]$}
    \psfrag{CF-mMIMO,Full-power123}{\scriptsize \hspace{-0.08cm} $\text{CF-mMIMO}_\text{Full power}$}
    \psfrag{CF-mMIMO,High-SINR}{\scriptsize \hspace{-0.08cm} $\text{CF-mMIMO}_\text{High-SINR}$}
    \psfrag{CF-mMIMO,WMMSE}{\scriptsize \hspace{-0.08cm} $\text{CF-mMIMO}_\text{W-MMSE}$}
    \psfrag{UC-mMIMO,Full-power}{\scriptsize \hspace{-0.08cm} $\text{UC-mMIMO}_\text{Full power}$}
    \psfrag{UC-mMIMO,High-SINR}{\scriptsize \hspace{-0.08cm} $\text{UC-mMIMO}_\text{High-SINR}$}
    \psfrag{UC-mMIMO,WMMSE}{\scriptsize \hspace{-0.08cm} $\text{UC-mMIMO}_\text{W-MMSE}$}
    }}
    \caption{The CDF of EE for $K\!=\!20$, $M\!=\!200$, and $A\!=\!4$. Herein, we assume clear weather and FSO-only fronthauling.}
    \label{Fig:Fig.7}
\end{figure}
Fig.~\ref{Fig:Fig.7} represents the cumulative distribution function (CDF) of the CF- and UC-mMIMO networks' EEs. For depicting this figure, we assume a clear weather condition and FSO-only fronthauling with $\Delta_a\!=\!0$. In both networks, full power allocation mechanisms are compared with the optimal ones based on the high-SINR and W-MMSE solutions. It is concluded that optimally allocating the UEs' and APs' transmission powers enhances the CF- and UC-mMIMO networks' performances with averagely $15\%$ and $8\%$ better EEs, respectively, in comparison to the full power allocations. It is also shown that applying the W-MMSE method, with sensitivity $\varepsilon\!=\!10^{-3}$, offers a bit better performance compared to the high-SINR approximation, at the cost of more complex calculations. Furthermore, the UC-mMIMO network outperforms the CF-mMIMO one with an average $83\%$ higher EE for both full and optimal power allocations. The main reason is that, in the CF-mMIMO network, the APs perform noisy estimations of far UEs' distorted channels and imperfectly decode their received data with low SINRs, unlikely in the UC-mMIMO network.
Fig.~\ref{Fig:Fig.8} is represented to investigate the cognitive fronthaul assignment proposed in Algorithm \ref{Alg:Alg.1} under different weather conditions for the CF- and UC-mMIMO networks. To this end, we consider four scenarios, wherein the numbers of FSO links with good, moderate, and poor alignments are dedicated for each scenario in Table~\ref{Table:Tab3}; scenarios A, B, C, and D. In Fig.~\ref{Fig:Fig.8}, optimal numbers of FSO-only, RF-only, and RF--FSO fronthaul links, for each weather condition and alignment scenario, are derived. It is verified that the RF links replace the FSO ones as the weather condition becomes unfavorable and the number of poorly-aligned links increases. For instance, it is suggested to deploy the fronthaul links using mostly the RF-only technology under the snowy and foggy states, even though the RF bandwidth is split between the access and fronthaul links such that the data rates drop. For intermediate conditions, RF--FSO links are preferred to enhance the networks' EEs. For this figure, the transmission powers at the UEs and APs are optimally allocated by applying the W-MMSE technique.
\noindent
\renewcommand{\arraystretch}{1}
\begin{table}[!t]
\begin{center}
\caption{FSO alignment scenarios with $M\!=\!200$.}\label{Table:Tab3}
\begin{tabular}{ | c | c | c | c || c | c | c | c | }
\hline
\small \!\!Scenario\!\! & \small \!Good\! & \small \!Mod.\! & \small \!Poor\! & \small \!\!Scenario\!\! & \small \!\!Good\!\! & \small \!Mod.\! & \small \!Poor\!\\
\hline
\hline
\small A & \footnotesize $20$ & \footnotesize $20$ & \footnotesize $160$ & \small C & \footnotesize $60$ & \footnotesize $60$ & \footnotesize $80$\\
\hline
\small B & \footnotesize $40$ & \footnotesize $40$ & \footnotesize $120$ & \small D & \footnotesize $80$ & \footnotesize $80$ & \footnotesize $40$\\
\hline
\end{tabular}
\medskip
\end{center}
\end{table}
\begin{figure}[t!]
    \centering
    \subfloat{
    \pstool[scale=0.6]{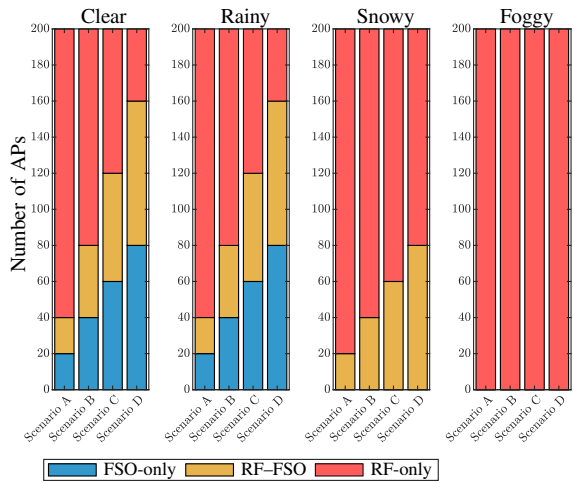}{
    \psfrag{NumberOfAPs}{\footnotesize \hspace{-0.35cm} Number of APs}
    \psfrag{RF/FSO123}{\scriptsize \hspace{-0.05cm} RF--FSO}
    \psfrag{RF1234567}{\scriptsize \hspace{-0.05cm} RF-only}
    \psfrag{FSO1234567}{\scriptsize \hspace{-0.05cm} FSO-only}
    \psfrag{Foggy}{\footnotesize \hspace{-0.13cm} Foggy}
    \psfrag{Snowy}{\footnotesize \hspace{-0.13cm} Snowy}
    \psfrag{Rainy}{\footnotesize \hspace{-0.13cm} Rainy}
    \psfrag{Clear}{\footnotesize \hspace{-0.13cm} Clear}
    }}
    \caption{The cognitive fronthaul assignment with different FSO alignment scenarios and weather conditions. Herein, $K\!=\!20$, $M\!=\!200$, and $A\!=\!4$. }
    \label{Fig:Fig.8}
\end{figure}

In Fig.~\ref{Fig:Fig.9}, the EE of the CF-mMIMO network is investigated and compared for various fronthaul assignment policies and scenario C; FSO-only, RF-only, RF\&FSO, and cognitive assignment. The FSO-only links are sensitive to weather conditions and do not offer high EE in snowy and foggy conditions. Even though the RF-only links are tolerable against adverse weather conditions, they have lower data rates than the FSO-only and RF\&FSO. Besides, the RF\&FSO policy decreases the EE of the network since the same data are sent in a parallel manner over both FSO and RF links with increased consumed powers, although it provides high sum-rates. Thus, applying the cognitive assignment boosts the network's performance for all weather conditions. Comparing with the RF\&FSO policy, the cognitive fronthaul assignment delivers $64\%$, $66\%$, $97\%$, and $198\%$ on average higher EE in the clear, rainy, snowy, and foggy conditions, respectively. Similarly, the results can be extended to other scenarios and also the UC-mMIMO network.
\begin{figure}[t!]
    \centering
    \subfloat{
    \pstool[scale=0.57]{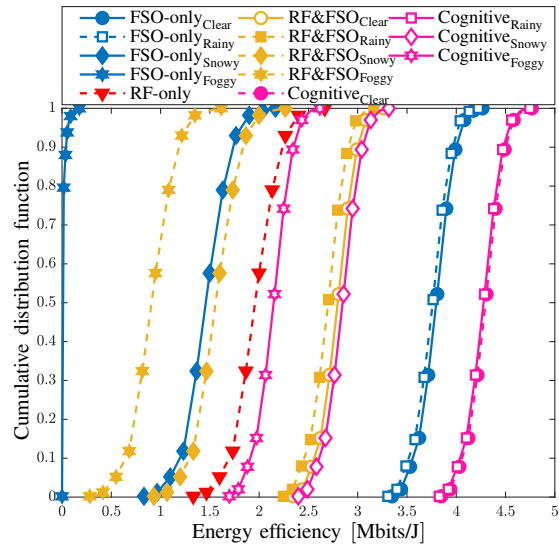}{
    \psfrag{Energyefficiency[Mbits/J]}{\footnotesize \hspace{-0.55cm} Energy efficiency $[\text{Mbits/J}]$}
    \psfrag{Cumulativedistributionfunction}{\footnotesize \hspace{-0.59cm} Cumulative distribution function}
    \psfrag{FSO-only,clear}{\scriptsize \hspace{-0.08cm} $\text{FSO-only}_\text{Clear}$}
    \psfrag{FSO-only,rainy}{\scriptsize \hspace{-0.08cm} $\text{FSO-only}_\text{Rainy}$}
    \psfrag{FSO-only,snowy}{\scriptsize \hspace{-0.08cm} $\text{FSO-only}_\text{Snowy}$}
    \psfrag{FSO-only,foggy}{\scriptsize \hspace{-0.08cm} $\text{FSO-only}_\text{Foggy}$}
    \psfrag{RF/FSO,clear}{\scriptsize \hspace{-0.08cm} $\text{RF\&FSO}_\text{Clear}$}
    \psfrag{RF/FSO,rainy}{\scriptsize \hspace{-0.08cm} $\text{RF\&FSO}_\text{Rainy}$}
    \psfrag{RF/FSO,snowy}{\scriptsize \hspace{-0.08cm} $\text{RF\&FSO}_\text{Snowy}$}
    \psfrag{RF/FSO,foggy}{\scriptsize \hspace{-0.08cm} $\text{RF\&FSO}_\text{Foggy}$}
    \psfrag{RF-only}{\scriptsize \hspace{-0.08cm} $\text{RF-only}$}
    \psfrag{Cognitive,clear}{\scriptsize \hspace{-0.08cm} $\text{Cognitive}_\text{Clear}$}
    \psfrag{Cognitive,rainy}{\scriptsize \hspace{-0.08cm} $\text{Cognitive}_\text{Rainy}$}
    \psfrag{Cognitive,snowy}{\scriptsize \hspace{-0.08cm} $\text{Cognitive}_\text{Snowy}$}
    \psfrag{Cognitive,foggy}{\scriptsize \hspace{-0.08cm} $\text{Cognitive}_\text{Foggy}$}
    }}
    \caption{The CDF of the CF-mMIMO network's sum-SE for different fronthaul assignment policies, wherein $K\!=\!20$, $M\!=\!200$, and $A\!=\!4$.}
    \label{Fig:Fig.9}
\end{figure}


\begin{figure}[t!]
    \centering
    \subfloat{
    \pstool[scale=0.57]{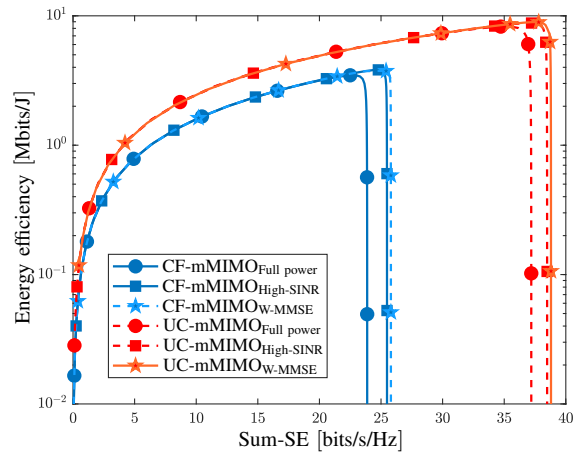}{
    \psfrag{SumSE[bits/s/Hz]}{\footnotesize \hspace{-0.4cm} Sum-SE $[\text{bits/s/Hz}]$}
    \psfrag{Energyefficiency[Mbits/J]}{\footnotesize \hspace{-0.52cm} Energy efficiency $[\text{Mbits/J}]$}
    \psfrag{CF-mMIMO,Full-power123}{\scriptsize \hspace{-0.08cm} $\text{CF-mMIMO}_\text{Full power}$}
    \psfrag{CF-mMIMO,High-SINR}{\scriptsize \hspace{-0.08cm} $\text{CF-mMIMO}_\text{High-SINR}$}
    \psfrag{CF-mMIMO,WMMSE}{\scriptsize \hspace{-0.08cm} $\text{CF-mMIMO}_\text{W-MMSE}$}
    \psfrag{UC-mMIMO,Full-power}{\scriptsize \hspace{-0.08cm} $\text{UC-mMIMO}_\text{Full power}$}
    \psfrag{UC-mMIMO,High-SINR}{\scriptsize \hspace{-0.08cm} $\text{UC-mMIMO}_\text{High-SINR}$}
    \psfrag{UC-mMIMO,WMMSE}{\scriptsize \hspace{-0.08cm} $\text{UC-mMIMO}_\text{W-MMSE}$}
    }}
    \caption{EE versus sum-SE for $K\!=\!20$, $M\!=\!200$, and $A\!=\!4$. Here, we assume clear weather and FSO-only fronthauling.}
    \label{Fig:Fig.10}
\end{figure}
The CF- and UC-mMIMO networks' EEs versus their sum spectral efficiencies (sum-SEs) are illustrated in Fig.~\ref{Fig:Fig.10}. For presenting this figure, we assume a clear weather condition and the FSO-only fronthauling with $\Delta_a\!=\!0$.
This figure verifies the results discussed in Fig.~\ref{Fig:Fig.7} for the CF- and UC-mMIMO networks with full and optimal power allocations based on the high-SINR and W-MMSE solutions. 

\begin{figure}[t!]
    \centering
    \subfloat[]{
    \pstool[scale=0.53]{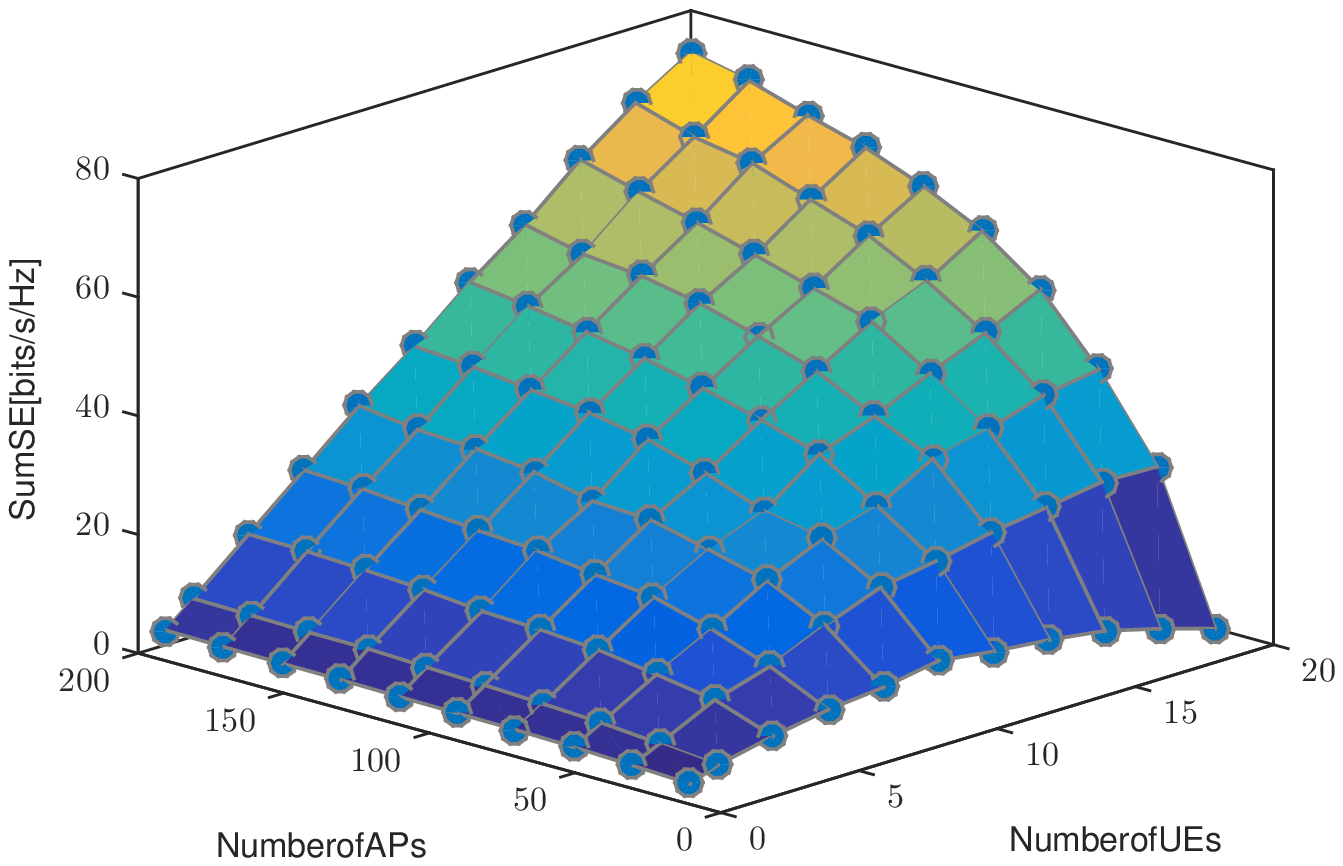}{
    \psfrag{SumSE[bits/s/Hz]}{\footnotesize \hspace{-0.6cm} Sum-SE $[\text{bits/s/Hz}]$}
    \psfrag{NumberofAPs}{\footnotesize \hspace{-0.5cm} Number of APs}
    \psfrag{NumberofUEs}{\footnotesize \hspace{-0.5cm} Number of UEs}
    \psfrag{CF-mMIMO,Full-power123}{\scriptsize \hspace{-0.08cm} $\text{CF-mMIMO}$}
    \psfrag{CF-mMIMO,High-SINR}{\scriptsize \hspace{-0.08cm} $\text{CF-mMIMO}_\text{High-SINR}$}
    \psfrag{CF-mMIMO,WMMSE}{\scriptsize \hspace{-0.08cm} $\text{CF-mMIMO}_\text{W-MMSE}$}
    \psfrag{UC-mMIMO,Full-power}{\scriptsize \hspace{-0.08cm} $\text{UC-mMIMO}$}
    \psfrag{UC-mMIMO,High-SINR}{\scriptsize \hspace{-0.08cm} $\text{UC-mMIMO}_\text{High-SINR}$}
    \psfrag{UC-mMIMO,WMMSE}{\scriptsize \hspace{-0.08cm} $\text{UC-mMIMO}_\text{W-MMSE}$}
    }}
    \hfill
    \subfloat[]{
    \pstool[scale=0.53]{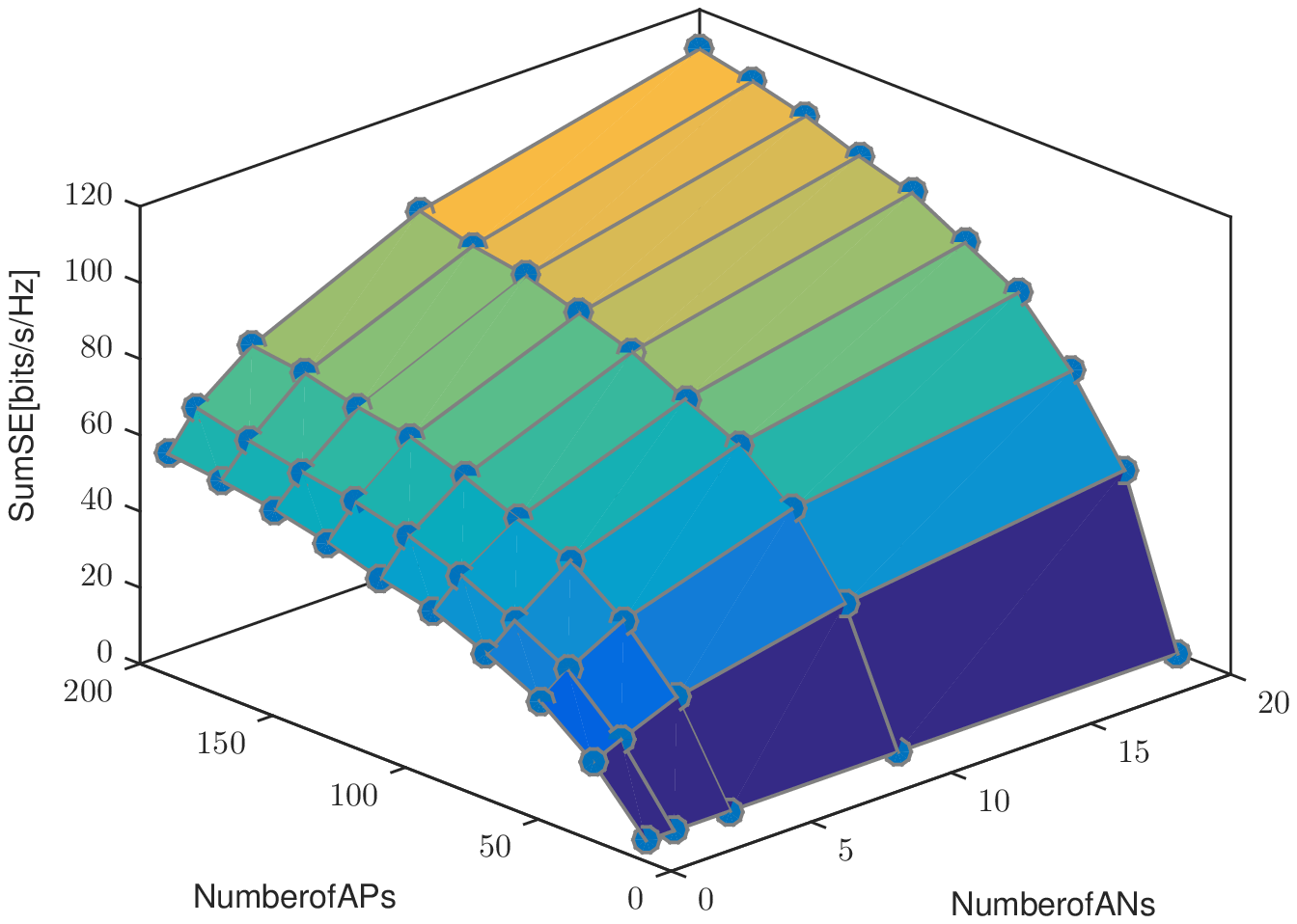}{
    \psfrag{SumSE[bits/s/Hz]}{\footnotesize \hspace{-0.6cm} Sum-SE $[\text{bits/s/Hz}]$}
    \psfrag{NumberofAPs}{\footnotesize \hspace{-0.5cm} Number of APs}
    \psfrag{NumberofANs}{\footnotesize \hspace{-0.5cm} Number of ANs}
    \psfrag{CF-mMIMO,Full-power123}{\scriptsize \hspace{-0.08cm} $\text{CF-mMIMO}$}
    \psfrag{CF-mMIMO,High-SINR}{\scriptsize \hspace{-0.08cm} $\text{CF-mMIMO}_\text{High-SINR}$}
    \psfrag{CF-mMIMO,WMMSE}{\scriptsize \hspace{-0.08cm} $\text{CF-mMIMO}_\text{W-MMSE}$}
    \psfrag{UC-mMIMO,Full-power}{\scriptsize \hspace{-0.08cm} $\text{UC-mMIMO}$}
    \psfrag{UC-mMIMO,High-SINR}{\scriptsize \hspace{-0.08cm} $\text{UC-mMIMO}_\text{High-SINR}$}
    \psfrag{UC-mMIMO,WMMSE}{\scriptsize \hspace{-0.08cm} $\text{UC-mMIMO}_\text{W-MMSE}$}
    }}
    \caption{Sum-SE vs the numbers of the (a) UEs and APs with $A\!=\!4$ and (b) APs and ANs with $K\!=\!20$.}
    \label{Fig:Fig.11}
\end{figure}
Fig.~\ref{Fig:Fig.11}\,(a) analyzes the CF-mMIMO network's performance for variable numbers of the UEs, i.e., $K$, and APs, i.e., $M$, with the fixed number of the ANs, i.e., $A$. It is presented that increasing the $K$ and $M$ continuously enhances the sum-SE to its highest point, then the growth rate declines. As an illustrative example, the sum-SE's growth rate at $M\!=\!60$ is $1.2$ times more than that at $M\!=\!120$ with $K\!=\!10$. Likewise, Fig.~\ref{Fig:Fig.11}\,(b) presents the sum-SE of the CF-mMIMO network versus the $M$ and $A$ for the fixed $K$. Even though the network's performance evolves by increasing the numbers of the APs and ANs, it is saturated for larger parameters. Thus, adding extra APs and ANs does not guarantee better network operation with high energy and cost efficiencies. The same conclusion is also hold for the UC-mMIMO network.

\section{Conclusion} \label{Sec:Sec7}
We studied the uplink of the CF- and the UC-mMIMO networks, wherein the APs are clustered and connected to corresponding ANs via their fronthaul links, and the ANs are connected to the CPU through their fiber backhaul links. For the uplink data transmission, all APs and ANs convert their received signals to be well-matched to the radio or optical links and transmit them to the other ends. After acquiring the CSI at the APs and ANs and sharing them with the CPU, uplink achievable data rates were derived for the CF- and UC-mMIMO networks by applying the MRC and UatF techniques. We formulated an optimization problem to maximize the CF- and UC-mMIMO networks' EEs via optimally allocating the transmission powers at the UEs and APs, subject to maximum transmission and consumed power. Two solutions based on the high-SINR and W-MMSE approaches were proposed to overcome the non-convexity of the optimization problem. Furthermore, a cognitive RF--FSO fronthaul assignment algorithm was suggested to enhance the CF- and UC-mMIMO networks' performances under FSO misalignment and adverse weather conditions. Finally, numerical results were represented to analyze and compare the performances of the CF- and UC-mMIMO networks. It was verified that the UC-mMIMO network overcomes the CF-mMIMO one with $83\%$ on average higher EE. The optimal power allocations also provide $15\%$ and $8\%$ increases in the CF- and UC-mMIMO networks' EEs, respectively. It was shown that even though the RF bandwidth is shared between the access and RF-only or RF--FSO fronthaul links, the cognitive assignment algorithm boosts the CF-mMIMO network's EE up to $198\%$ in unfavorable conditions, compared to the one with FSO-only, RF-only, or RF\&FSO fronthaul links.

\appendices

\section{SINR Derivation}\label{App:App.1}
The SINR of the $k$th UE is defined as 
\begin{equation}\label{Eq:Eq41}
    \text{SINR}_k = \dfrac{|\text{DS}_k|^2}
    {\mathbb{E}\big\{|\text{BU}_k|^2\big\} +\sum\limits_{\substack{k^\prime=1 \\ k^\prime\neq k}}^{K} \mathbb{E}\big\{|\text{IUI}_{k k^\prime}|^2\big\} +\mathbb{E}\big\{|\text{N}_{k}|^2\big\}}.
\end{equation}
After a sequence of mathematical manipulations, the expressions of all terms given in (\ref{Eq:Eq41}) are derived in the following parts.  

\subsection{$\normalfont{\text{DS}}_k$} We consider $\mathbf{g}_{mk} \!=\! \hat{\mathbf{g}}_{mk} \!+\! \mathbf{e}_{mk}$ and $I_{am}^\text{RF} \!=\!  \hat{I}_{am}^\text{RF} \!+\! {e}_{am}$, with $\mathbf{e}_{mk}$ and ${e}_{am}$ representing the estimation errors modeled as zero-mean Gaussian RVs with the variances of $\beta_{mk}-\gamma_{mk}$ and $\beta_{am}-\gamma_{am}^\text{RF}$, respectively, where $\mathbf{e}_{mk} \!\perp\! \hat{\mathbf{g}}_{mk}$ and ${e}_{am} \!\perp\! \hat{I}_{am}^\text{RF}$. It is assumed that the FSO channels are perfectly estimated, i.e., $\Gamma_{am}^2\!=\!\gamma_{am}^\text{FSO}$. However, since the channel coefficients are i.i.d., we attain
\begin{align}\label{Eq:EqAp1} \nonumber
    \text{DS}_k 
    &= \sqrt{\eta_k} \Bigg(\! \sum\limits_{a=1}^{A} \sum\limits_{m\in \mathcal{M}_k(a)} \!\! \epsilon_{am} J_{am} \mathbb{E}\Big\{\!(\hat{I}_{am}^\text{FSO})^*I_{am}^\text{FSO}\!\Big\} \mathbb{E}\Big\{\!\hat{\mathbf{g}}_{mk}^H \mathbf{g}_{mk}\!\Big\}\\ \nonumber
    &~~~+ \sum\limits_{a=1}^{A} \sum\limits_{m\in \mathcal{M}_k(a)} \!\! \epsilon_{am}^\prime J_{am}^\prime \mathbb{E}\Big\{\!(\hat{I}_{am}^\text{RF})^*I_{am}^\text{RF}\!\Big\} \mathbb{E}\Big\{\!\hat{\mathbf{g}}_{mk}^H \mathbf{g}_{mk}\Big\}\!\!\Bigg) \\ 
    &=\sqrt{\eta_k}\sum\limits_{a=1}^{A} \sum\limits_{m\in \mathcal{M}_k(a)} \!\! \bigg(\! \epsilon_{am} J_{am} \gamma_{am}^\text{FSO} +  \epsilon_{am}^\prime J_{am}^\prime \gamma_{am}^\text{RF} \! \bigg) \gamma_{mk}.
\end{align}

\subsection{ $\mathbb{E}\big\{|\normalfont{\text{BU}}_k|^2\big\}$} Since the variance of a sum of independent random variables is equal to the sum of the variances, we have
\begin{align}\label{Eq:EqAp2} \nonumber
   &\mathbb{E}\big\{|\text{BU}_k|^2\big\} = 
    \eta_k \sum\limits_{a=1}^{A} \sum\limits_{m\in \mathcal{M}_k(a)} \!\! \Bigg(\! \epsilon_{am}^2 J_{am}^2 \\ \nonumber
    &~~~\times\!
    \mathbb{E} \Bigg\{\!\bigg| (\hat{I}_{am}^\text{FSO})^* {I}_{am}^\text{FSO} \hat{\mathbf{g}}_{mk}^H {\mathbf{g}}_{mk} \!-\! \mathbb{E}\bigg\{\!(\hat{I}_{am}^\text{FSO})^* {I}_{am}^\text{FSO} \hat{\mathbf{g}}_{mk}^H {\mathbf{g}}_{mk}\bigg\} \bigg|^2\!\Bigg\}\\ \nonumber
    &~~~+ \epsilon_{am}^{\prime 2} J_{am}^{\prime 2}\\ \nonumber
    &\times\!\mathbb{E} \Bigg\{\!\bigg| (\hat{I}_{am}^\text{RF})^* {I}_{am}^\text{RF} \hat{\mathbf{g}}_{mk}^H {\mathbf{g}}_{mk} \!-\! \mathbb{E}\bigg\{\!(\hat{I}_{am}^\text{RF})^* {I}_{am}^\text{RF} \hat{\mathbf{g}}_{mk}^H {\mathbf{g}}_{mk}\!\bigg\} \bigg|^2\!\Bigg\}\!\!\Bigg) \\ \nonumber
    &= \eta_k \sum\limits_{a=1}^{A} \sum\limits_{m\in \mathcal{M}_k(a)} \!\!\! \Bigg(\! \epsilon_{am}^2 J_{am}^2 \\ \nonumber
    &~~~\times\!\Bigg[ \mathbb{E} \bigg\{\! \big|\hat{I}_{am}^\text{FSO}\big|^4\!\bigg\} \mathbb{E} \bigg\{\!\big|\hat{\mathbf{g}}_{mk}^H \big(\hat{\mathbf{g}}_{mk} \!+\! \mathbf{e}_{mk}\big)\big|^2\!\bigg\}\\ \nonumber
    &~~~- \bigg| \mathbb{E}\bigg\{\!\big|\hat{I}_{am}^\text{FSO}\big|^2\!\bigg\} \mathbb{E}\bigg\{\!\hat{\mathbf{g}}_{mk}^H \big(\hat{\mathbf{g}}_{mk} \!+\! \mathbf{e}_{mk}\big)\!\bigg\} \bigg|^2 \Bigg] + \epsilon_{am}^{\prime 2} J_{am}^{\prime 2} \\ \nonumber
    &~~~\times\!\Bigg[ \mathbb{E} \bigg\{\! \big|(\hat{I}_{am}^\text{RF})^* \big(\hat{I}_{am}^\text{RF} \!+\!e_{am}\big)\big|^2\!\bigg\} \mathbb{E} \bigg\{\!\big|\hat{\mathbf{g}}_{mk}^H \big(\hat{\mathbf{g}}_{mk} \!+\! \mathbf{e}_{mk}\big)\big|^2\!\bigg\}\\ \nonumber
    &~~~- \bigg|\mathbb{E} \bigg\{\! (\hat{I}_{am}^\text{RF})^* \big(\hat{I}_{am}^\text{RF} \!+\!e_{am}\big)\!\bigg\} \mathbb{E}\bigg\{\!\hat{\mathbf{g}}_{mk}^H \big(\hat{\mathbf{g}}_{mk} \!+\! \mathbf{e}_{mk}\big)\!\bigg\} \bigg|^2 \Bigg]\Bigg) \\ \nonumber
    &= \eta_k \sum\limits_{a=1}^{A} \sum\limits_{m\in \mathcal{M}_k(a)} \!\! \gamma_{mk} \Bigg(\! \epsilon_{am}^2 J_{am}^2 (\gamma_{am}^\text{FSO})^2 \Big( \gamma_{mk} \!+\! 2\beta_{mk} \Big)\\
    &~~~ + \epsilon_{am}^{\prime 2} J_{am}^{\prime 2} \gamma_{am}^\text{RF} 
    \Big( \gamma_{am}^\text{RF}\beta_{mk} \!+\! \beta_{am} \gamma_{mk} \!+\! \beta_{am} \beta_{mk}\Big) \!\Bigg)\!.
\end{align}

\subsection{ $\mathbb{E}\big\{|\normalfont{\text{IUI}}_{kk^\prime}|^2\big\}$}
Since the pilot sequences are mutually orthogonal, we have 
\begin{align}\label{Eq:EqAp3} \nonumber
    &\mathbb{E}\big\{|\text{IUI}_{kk^\prime}|^2\big\} = \\ \nonumber
    &\eta_{k^\prime} \Bigg(\! 2 \sum\limits_{a=1}^{A} \sum\limits_{m\in \mathcal{M}_k(a)} \!\!\! \epsilon_{am}^2  J_{am}^2 (\gamma_{am}^\text{FSO})^2 \mathbb{E}\Bigg\{\!\bigg| \hat{\mathbf{g}}_{mk}^H \mathbf{g}_{mk^\prime}\bigg|^2\!\Bigg\}\\ \nonumber
    &+ \sum\limits_{a=1}^{A} \sum\limits_{m\in \mathcal{M}_k(a)}\!\! \epsilon_{am}^{\prime 2} J_{am}^{\prime 2} \gamma_{am}^\text{RF}\big(\gamma_{am}^\text{RF} \!+\! \beta_{am}\big)  \mathbb{E}\Bigg\{\!\bigg| \hat{\mathbf{g}}_{mk}^H \mathbf{g}_{mk^\prime} \bigg|^2\!\Bigg\}\!\Bigg)\\ \nonumber
    &=\eta_{k^\prime} \sum\limits_{a=1}^{A} \sum\limits_{m\in \mathcal{M}_k(a)} \!\! \bigg(\! 2\epsilon_{am}^2  J_{am}^2 (\gamma_{am}^\text{FSO})^2\\ \nonumber
    &+ \epsilon_{am}^{\prime 2} J_{am}^{\prime 2} \gamma_{am}^\text{RF}\big(\gamma_{am}^\text{RF} \!+\! \beta_{am}\big) \!\bigg) \zeta_{mk}^2 \Bigg[ \mathbb{E}\Bigg\{\!\bigg|\bm{\varphi}_k^H \bm{\omega}_{p,m}^H \mathbf{g}_{mk^\prime} \bigg|^2\!\Bigg\}\\ \nonumber
    &+ \mathbb{E}\Bigg\{\!\bigg| \sum\limits_{k^{\prime\prime}=1}^{K}  \bm{\varphi}_k^H \bm{\varphi}_{k^{\prime\prime}} \mathbf{g}_{mk^{\prime\prime}}^H \mathbf{g}_{mk^\prime}\bigg|^2\!\Bigg\} \\ \nonumber
    &=\eta_{k^\prime} \sum\limits_{a=1}^{A} \sum\limits_{m\in \mathcal{M}_k(a)} \!\! \beta_{mk^\prime} \bigg(\! 2\epsilon_{am}^2  J_{am}^2 (\gamma_{am}^\text{FSO})^2\\
    &+ \epsilon_{am}^{\prime 2} J_{am}^{\prime 2} \gamma_{am}^\text{RF}\big(\gamma_{am}^\text{RF} \!+\! \beta_{am}\big) \!\bigg)  \bigg(\!  \zeta_{mk}^2 \sigma_{p,m}^2  \!+\! \frac{\gamma_{mk}^2}{\beta_{mk}}\!\bigg).
\end{align}

\subsection{ $\mathbb{E}\big\{|\normalfont{\text{N}}_k|^2\big\}$} Because the additive noise consists of the i.i.d. elements, it can be shown that 
\begin{align}\label{Eq:EqAp4} \nonumber
    &\mathbb{E}\big\{|\text{N}_k|^2\big\}
    =  \sum\limits_{a=1}^{A} \sum\limits_{m\in \mathcal{M}_k(a)} \!\! \epsilon_{am}^2 \gamma_{am}^\text{FSO} \mathbb{E} \bigg\{\! \Big| \hat{\mathbf{g}}_{mk}^H \bm{\Theta}_{u,am}\Big|^2 \! \bigg\}\\ \nonumber
    &\hspace{1.35cm}~+ \sum\limits_{a=1}^{A} \sum\limits_{m\in \mathcal{M}_k(a)}\!\! \epsilon_{am}^{\prime 2} \gamma_{am}^\text{RF} \mathbb{E} \bigg\{\! \Big| \hat{\mathbf{g}}_{mk}^H \bm{\Theta}_{u,am}^\prime \Big|^2 \! \bigg\} \\ 
    &=  \sum\limits_{a=1}^{A} \sum\limits_{m\in \mathcal{M}_k(a)} \!\! \gamma_{mk} \Bigg(\! \epsilon_{am}^2 \gamma_{am}^\text{FSO}  {\Omega}_{u,am}^{2}\!+\! \epsilon_{am}^{\prime 2} \gamma_{am}^\text{RF}  {\Omega}_{u,am}^{\prime2} \!\Bigg)\!.
\end{align}

By inserting (\ref{Eq:EqAp1}), (\ref{Eq:EqAp2}), (\ref{Eq:EqAp3}), and (\ref{Eq:EqAp4}) into (\ref{Eq:Eq41}), with further mathematical manipulations, (\ref{Eq:Eq42}) is obtained.

\section{Mean-Square Error Computation}\label{App:App.2}
The desired signal $s_k$ is decoded by applying a beamforming coefficient $u_k$, given below
\begin{align}\label{Eq:EqEqAp5}
    \hat{s}_k = u_k r_{u,k}.
\end{align}
Therefore, the mean-square error of the decoding is computed as follows
\begin{align}\label{Eq:EqEqAp6} \nonumber
    &e_k=\mathbb{E}\Big\{\big|\hat{s}_k - s_k \big|^2 \Big\}=\mathbb{E}\Big\{\big|u_kr_{u,k}- s_k \big|^2 \Big\}=\\\nonumber
    &~\mathbb{E}\Bigg\{ \bigg| u_k \sqrt{\eta_k} \sum\limits_{a=1}^{A} \sum\limits_{m\in \mathcal{M}_k(a)}\!\! \epsilon_{am} J_{am} (\hat{I}_{am}^{\text{FSO}})^* I_{am}^{\text{FSO}} \hat{\mathbf{g}}_{mk}^H \mathbf{g}_{mk}\\ \nonumber
    &+ u_k \sqrt{\eta_k}\sum\limits_{a=1}^{A} \sum\limits_{m\in \mathcal{M}_k(a)}\!\! \epsilon_{am}^\prime J_{am}^\prime (\hat{I}_{am}^{\text{RF}})^* I_{am}^{\text{RF}} \hat{\mathbf{g}}_{mk}^H \mathbf{g}_{mk} -1 \bigg|^2 \Bigg\}  \\ \nonumber
    &+ |u_k|^2 \mathbb{E}\Bigg\{ \bigg| \sum\limits_{a=1}^{A} \sum\limits_{m\in \mathcal{M}_k(a)}\!\! \bigg(\!\epsilon_{am} J_{am}  (\hat{I}_{am}^{\text{FSO}})^* I_{am}^{\text{FSO}}\\ \nonumber
    &+ \epsilon_{am}^\prime J_{am}^\prime (\hat{I}_{am}^{\text{RF}})^* I_{am}^{\text{RF}} \!\bigg) \sum\limits_{k^\prime\neq k}^{K} \sqrt{\eta_{k^\prime}} \hat{\mathbf{g}}_{mk}^H \mathbf{g}_{mk^\prime} \bigg|^2 \Bigg\} \\  \nonumber
    &+ |u_k|^2 \mathbb{E}\Bigg\{ \bigg| \sum\limits_{a=1}^{A} \sum\limits_{m\in \mathcal{M}_k(a)}\!\! \epsilon_{am}  (\hat{I}_{am}^{\text{FSO}})^* \hat{\mathbf{g}}_{mk}^H \bm{\Theta}_{u,am} \bigg|^2 \Bigg\}\\
    &+ |u_k|^2 \mathbb{E}\Bigg\{ \bigg| \sum\limits_{a=1}^{A} \sum\limits_{m\in \mathcal{M}_k(a)}\!\! \epsilon_{am}^\prime  (\hat{I}_{am}^{\text{RF}})^* \hat{\mathbf{g}}_{mk}^H \bm{\Theta}_{u,am}^\prime \bigg|^2 \Bigg\}.
\end{align}
After some algebraic manipulations, (\ref{Eq:EqWMMSE2}) is obtained.

\balance
\bibliographystyle{IEEEtran}
\bibliography{References.bib}

\end{document}